\documentclass[12pt]{iopart}

\usepackage{color}
\usepackage{graphicx}
\usepackage{psfrag}

\usepackage{epsf}
\usepackage{epic}
\usepackage{eepic}

\usepackage{epsfig}
\usepackage{wrapfig}
\usepackage{amssymb, amsthm}
\usepackage{mathrsfs}
\usepackage{cite}
\usepackage{deleq}
\usepackage{url}
\usepackage{fancybox}

\usepackage{bigint}
\usepackage{mhequ}

{\catcode `\@=11 \global\let\AddToReset=\@addtoreset}
\AddToReset{equation}{section}

{\catcode `\@=11 \global\let\AddToReset=\@addtoreset}
\AddToReset{figure}{section}

\renewcommand{\theequation}{\thesection.\arabic{equation}}

\theoremstyle{plain}
\newtheorem{Theorem}{Theorem}[section]
\newtheorem{Corollary}[Theorem]{Corollary}
\newtheorem{Proposition}[Theorem]{Proposition}
\newtheorem{Lemma}[Theorem]{Lemma}
\newtheorem{cocoExa}[Theorem]{Example}

\newtheorem{Definition}[Theorem]{Definition}

\newtheorem{Remark}[Theorem]{{Remark}}

\newcommand{\bel}[1]{\begin{equation}\label{#1}}
\newcommand\ee{\end{equation}}

\newcommand{\bea}{\begin{eqnarray}}
\newcommand{\bean}{\begin{eqnarray}\nonumber}

\newcommand{\beaa}{\begin{eqnarray*}}
\newcommand{\eeaa}{\end{eqnarray*}}

\newcommand{\beal}[1]{\begin{eqnarray}\label{#1}}
\newcommand{\eea}{\end{eqnarray}}
\newcommand{\eeal}[1]{\label{#1}\end{eqnarray}}

\newcommand{\eq}[1]{(\ref{#1})}

\newcommand{\dotI}{{\partial I}}
\newcommand{\dotJ}{{\partial J}}
\newcommand{\dotcI}{{\partial\check{I}}{}}
\newcommand\mycal{\mathscr}
\newcommand{\mcM}{{\mycal M}}
\newcommand{\mcE}{{\mycal E}}
\newcommand{\mcO}{{\mycal O}}
\newcommand{\mcU}{{\mycal U}}
\newcommand{\mcV}{{\mycal V}}
\newcommand{\mcW}{{\mycal W}}
\newcommand{\mcD}{{\mycal D}}
\newcommand{\mcB}{{\mycal B}}
\newcommand{\mcDpSc}{\mcD{}^+_{J}(\hyp)}
\newcommand{\mcDmSc}{\mcD{}^-_{J}(\hyp)}
\newcommand{\mcDSpI}{{\mcD_{I}^+(\hyp)}}
\newcommand{\mcDSI}{{\mcD_{I}(\hyp)}}
\newcommand{\mcDmScg}{\mcD{}^-_{J,g}(\hyp)}
\newcommand{\mcDpS}{\mcD{}^+(\hyp)}
\newcommand{\mcDSmI}{{\mcD_{I}^-(\hyp)}}
\newcommand{\mcDSpIg}{{\mcD_{I,g}^+(\hyp)}}
\newcommand{\hyp}{\mycal S}

\newcommand\N{\ensuremath{\mathbb{N}}}
\newcommand\R{\ensuremath{\mathbb{R}}}
\newcommand{\chg}{{\check g}}
\newcommand{\hg}{{\hat g}}
\newcommand{\cI}{{\check{I}{}}}
\newcommand{\cJ}{{\check{J}{}}}
\newcommand{\loc}{\textrm{\scriptsize\upshape loc}}
\newcommand{\dcI}{\dotcI}
\newcommand{\eps}{{\epsilon}}
\newcommand{\hgkn}{\hat g_{k(n)}}

\newcommand{\distb}{\mathrm{dist}_{h}}

\newcommand{\mcDSc}{\mcD_{J}(\hyp)}
\newcommand{\mcDpScg}{\mcD{}^+_{J,g}(\hyp)}
\newcommand{\mcDScg}{\mcD_{J,g}(\hyp)}

\begin{document}
\title{On Lorentzian causality with continuous metrics}
\author{Piotr T Chru\'{s}ciel$^{1, 2}$, James D E Grant$^2$}%
\address{$^1$ Institut des Hautes \'Etudes Scientifiques, Bures-sur-Yvette}
\address{$^2$ Gravitational Physics, University of Vienna.}
\eads{\mailto{piotr.chrusciel@univie.ac.at}, \mailto{james.grant@univie.ac.at}}
\date{November 29, 2011}

\begin{abstract}
We present a systematic study of causality theory on Lorentzian manifolds with
continuous metrics.
Examples are given which show that some standard facts in smooth Lorentzian geometry, such as light-cones being hypersurfaces, are wrong when metrics which are merely continuous are considered. We show that existence of time functions remains true on domains of dependence with continuous metrics, and that $C^{0,1}$ differentiability of the metric suffices for many key results of the smooth causality theory.
\end{abstract}
\maketitle

\section{Causality for continuous metrics}
\label{s13III11.1}
\setcounter{equation}{0}

One of the factors that constrains the differentiability
requirements of the proof of the celebrated
Choquet-Bruhat--Geroch theorem~\cite{ChBGeroch}, of existence
and uniqueness of maximal globally hyperbolic vacuum
developments of general relativistic initial data, is the
degree of differentiability needed to carry out the Lorentzian
causality arguments that arise. Here one should keep in mind
that classical local existence and uniqueness of solutions of
vacuum Einstein equations in dimension $3+1$ applies to initial
data $(g, K)$ in the product of Sobolev spaces $H^s \times
H^{s-1}$ for $s>5/2$, and that the recent studies of
the Einstein equations~\cite{KlainermanRodnianski:r1,WangCones,WangRicci,%
KRS,Choquet-Bruhat:safari,Maxwell:Rough,Maxwell:Compact} assume even less differentiability.
On the other hand, the standard references on causality
seem to assume smoothness of the metric~\cite{HE,Beem-Ehrlich:Lorentz2,ONeill83,MinguzziSanchez,PenroseDiffTopo,Kriele},
while the presentation in~\cite{Senovilla,ChCausalityv1} requires $C^2$ metrics.%
\footnote{See also~\cite{Senovilla}, where some of the issues involved in trying to prove the singularity theorems for metric below $C^2$ regularity are discussed.}
Hence the need to revisit the causality theory for Lorentzian
metrics which are merely assumed to be continuous. Surprisingly enough, some standard
facts of the $C^2$ theory fail to hold for metrics with lower differentiability.
For example, we will show that the following statements are wrong:

\begin{enumerate}
 \item\label{Point1} light-cones are topological hypersurfaces of codimension one;
 \item\label{Point2} a piecewise differentiable causal curve which is not null
everywhere can be deformed, with end points fixed, to a timelike curve.
\end{enumerate}

Concerning point~(\ref{Point1}) above, we exhibit metrics with
light-cones which have non-empty interior.

Researchers familiar with Lorentzian geometry will
recognize point~(\ref{Point2}) as an essential tool in many arguments.
One needs therefore to reexamine the corresponding results, documenting their failure or
finding a replacement for the proof.

In the course of the analysis, one is naturally led to the notion of a
\emph{causal bubble\/} which, roughly speaking, is defined to be an open set which
can be reached from, say a point $p$, by causal curves but not by timelike
ones.

The object of this paper is to present
the above, reassessing that part of causality theory which has been presented
in~\cite{ChCausalityv1} from the perspective of continuous metrics. One of our main results is the proof
that domains of dependence equipped with continuous metrics continue to admit Cauchy time
functions.%
\footnote{Once the first draft of this paper was completed (arXiv:1111.0400v1) we were informed of~\cite{FathiSiconolfiTime}, where
the result is proved by completely different methods, and in greater generality. Subsequent email exchanges with
A.~Fathi inspired the proof of Theorem~\ref{T9IX11.1} below.}
Another key result is the observation that the causality theory
developed in~\cite{ChCausalityv1} for $C^2$ metrics remains true for $C^{0,1}$ metrics.

An application of our work to the general relativistic Cauchy problem can be found in~\cite{ChMGHD}. In fact, the main motivation for the current work was to develop the Lorentzian-causality tools needed for that last reference. The existing treatments of continuous Lorentzian metrics known to us (see, e.g., the references in Section~2 of \cite{AGPS}) are, unfortunately, not well-suited to the study of this particular application.

The conventions and notations of~\cite{ChCausalityv1} are used throughout. In particular all manifolds are assumed to be connected, Hausdorff, and paracompact. As we are
interested in $C^0$ metrics, the natural associated differentiability class of
the manifolds is $C^1$. Now, $C^1$ manifolds always possess a $C^\infty$
subatlas, and we will choose some such subatlas   whenever convenient. For instance, when a smooth nearby metric is needed,  we choose some smooth subatlas, obtaining thus a smooth manifold which is $C^1$-diffeomorphic to the original one.  We  carry out the smoothing construction on this new manifold, obtain the desired conclusions there, and return to the original atlas at the end of the argument.

A \emph{space-time}, throughout, will mean a time-oriented Lorentzian manifold $(\mcM, g)$.

\subsection{Some background on manifolds}
 \label{ss3IV11.1}

Let $\mcM$ be an $n$-dimensional smooth manifold. By this, we will mean that we have a maximal atlas $\mathcal{A} = \left\{ \left( V_{\alpha}, \varphi_{\alpha} \right) : \alpha \in A \right\}$ of charts, each of which consists of an open set $V_{\alpha} \subseteq \mcM$ and a bijection
 $\varphi_{\alpha} \colon V_{\alpha} \to \R^n$, where $\varphi_{\alpha} \left( V_{\alpha} \right)$ is an open subset of $\R^n$. These charts are compatible in the sense that
\begin{itemize}
\item[$\bullet$] For all $\alpha, \beta \in A$ such that $V_{\alpha} \cap V_{\beta} \neq \emptyset$, the sets $\varphi_{\alpha}\left( V_{\alpha} \cap V_{\beta} \right)$ and $\varphi_{\beta}\left( V_{\alpha} \cap V_{\beta} \right)$ are open subsets of $\R^n$;
\item[$\bullet$] For all $\alpha, \beta \in A$ such that $V_{\alpha} \cap V_{\beta} \neq \emptyset$, the maps
\begin{equation}
\label{transition}
\varphi_{\alpha} \circ \varphi_{\beta}^{-1} \colon \varphi_{\beta}\left( V_{\alpha} \cap V_{\beta} \right) \to \varphi_{\alpha}\left( V_{\alpha} \cap V_{\beta} \right)
\end{equation}
are $C^{\infty}$.
\end{itemize}

The collection of subsets $\mathcal{B} := \left\{ V_{\alpha} : \alpha \in A \right\}$ defines a basis for a topology on $\mcM$ with respect to which the maps $\varphi_\alpha \colon\, V_{\alpha} \to \R^n$ are continuous. The manifold $\mcM$, with this topology, is automatically locally compact, i.e., any point has a compact neighbourhood. Note that we impose that the transition maps~(\ref{transition}) are $C^{\infty}$ since we will later wish to approximate tensorial objects on $\mcM$ by corresponding smooth objects. In reality, imposing that the transition functions are $C^3$ would be sufficient for most of our considerations.

We impose the additional topological conditions that the manifold $\mcM$ be connected, Hausdorff and paracompact. The fact that $\mcM$ is Hausdorff implies that, in addition to being locally compact, $\mcM$ has the property that any point $p \in \mcM$ has an open neighbourhood with compact closure.\footnote{Or, equivalently, that every point has a compact, \emph{closed\/} neighbourhood.}

\begin{Remark}
{\rm
A theorem of Geroch~\cite{GerochSpinI} shows that a manifold with a
$C^2$ Lorentzian metric is necessarily paracompact.
However, Geroch's construction requires extensive use of the exponential map,
and therefore this method cannot be applied when the metric is merely continuous.
As such, we impose the condition that $\mcM$ be paracompact by hand.
}
\end{Remark}

Since $\mcM$ is Hausdorff and paracompact, it follows that $\mcM$ admits smooth partitions of unity. In particular, $\mcM$ admits a smooth Riemannian metric. Let us once and for all choose a Riemannian metric, say $h$, on $\mcM$, as differentiable as the atlas allows. Without loss of generality~\cite{NomizuOzeki}, we will assume that this metric is complete.

\subsection{Causality: basic notions}

Somewhat surprisingly, the causality theory for continuous
metrics appears to present various difficulties. For instance,
following~\cite{ChCausalityv1}, it is tempting to continue to
define a \emph{timelike curve\/} as a locally Lipschitz curve with tangent vector which
is timelike almost everywhere. With this definition, for
metrics which are only continuous, it is not even clear whether
or not the timelike futures remain open, and in fact we do not know the answer to this question.
Next, one would like to
keep the definition of a \emph{causal curve\/} as a
locally Lipschitz curve with tangent vector which is causal
almost everywhere. With this definition, or for that matter with any alternative,
it is not clear
that a pointwise limit of causal curves is causal. To settle this last question, and some others, we will extensively use families of smooth metrics which
approach the given metric $g$ in a specific way. For this, some
notation will be useful. Let $\chg$ be a Lorentzian metric; we
will say that the light-cones of $\chg$ are \emph{strictly
smaller\/} than those of $g$, or that those of $g$ are
\emph{strictly larger\/} than those of $\chg$, and we shall write
\bel{13III11.0}
 \chg \prec g
  \;,
\ee
if it is true that
\bel{13III11.1}
 \chg(T,T) \le 0, \quad T \neq 0 \quad
  \Longrightarrow
  \quad
  g(T,T)<0
  \;.
\ee

We claim that

\begin{Proposition}
 \label{P13III11.1}
  Let $(\mcM,g)$ be a space-time with a continuous metric $g$.
 For every $\epsilon>0$ there exist \emph{smooth\/} Lorentzian
 metrics $\chg$ and $\hg$ such that
\bel{13III11.2}
 \chg \prec g \prec \hg
\ee
and such that
\bel{13III11.3}
 d(\chg , g)+ d( \hg ,g)\le \epsilon
 \;,
\ee
where
\bel{13III11.4}
 d(g_1, g_2)= \sup_{0\ne X,Y \in TM}
 \frac{|g_1(X,Y)-g_2(X,Y)|}{|X|_h |Y|_h}
 \;.
\ee
(Recall that $h$ denotes a fixed smooth, complete Riemannian metric.)
\end{Proposition}

\proof
By definition of a space-time, there exists on $\mcM$ a continuous timelike vector field
 $T$. (In fact, a timelike vector field such that $T\otimes T$ is continuous would suffice for the argument below.) Let
$\{B_i(3r_i)\}_{i\in \N}$ be a family of coordinate balls of
radius $3r_i$ such that the balls $\{B_i(r_i)\}_{i\in \N}$ cover $\mcM$.
Let $n+1$ denote the dimension of $\mcM$, and let $\varphi$ be a smooth non-negative radial
function on $\R^{n+1}$ supported in the unit ball, with
integral one, set $\varphi_\eta(x) := \eta^{-(n+1)} \varphi(x/\eta)$.
Let $\{\chi_i\}_{i\in\N}$ be a partition of unity subordinate to the cover, thus $\chi_i$ is supported in $B_i(r_i)$, only a finite number of the $\chi_i$'s are non-zero in a neighborhood of each point, and they sum to one.
For $\eta_i< r_i$, in local coordinates on $B_i(r_i)$ set
\beal{25III12.15}
\varphi_\eta \star g_{\mu\nu}
  &   = &
  \sum_j \chi_ j \varphi_{\eta } \star
    g_{\mu\nu}
     \;,
\qquad
 T(\eta)_\mu
  =
   \sum_j \chi_j\varphi_{\eta } \star (g_{\mu\nu}T ^\nu)
   \;,
\\
 g(\eta,\lambda) _{\mu\nu}
  &  = &
   \varphi_\eta \star g_{\mu\nu} +\lambda T(\eta)_\mu T(\eta)_\nu
  \;,
\eeal{25III12.16}
where $\star$ denotes the usual convolution on $\R^{n+1}$. Then
the $g(\eta,\lambda)$ are smooth tensor fields on $B_i(r_i)$ and converge
uniformly to $g$ on $B_i(r_i)$ as $\eta$ and $\lambda$ tend to
zero. Similarly the vector fields $T(\eta)$ are smooth and
converge uniformly to $T$ on any compact subset of $\mcM$. In particular the vector
fields $T(\eta)$ are timelike on $B_i(r_i)$ both for $g$ and $g(\eta,0)$ for
all $\eta$ sufficiently small, where the notion of ``sufficiently small" might depend upon $i$.
%
%
There exist constants $\eta_i,c_i>0$ such that we have
$|T(\eta)(X)|\ge c_i$ for all $g$-causal vectors $X$, with
$h$--length equal to one, defined over $B_i(r_i)$, and for all
$0\le |\eta| \le \eta_i$. Given any $0< |\lambda|\le 1$ we can
choose $0<\eta(\lambda)\le \eta_i$ so that for all $0\le
\eta\le \eta(\lambda)$ we have
$$|(\varphi_\eta \star g - g)
 (X,X) |\le |\lambda| c_i^2/2
$$
for all such $X$. It
then follows, for all  $g$-causal $X$ over $B_i(r_i)$, for
$\lambda<0$ and for $\eta=\eta(\lambda)$, that
$$
 g(\eta,\lambda)(X,X)=\underbrace{g(X,X)}_{\le 0} + \underbrace{(\varphi_\eta \star g - g)
 (X,X)}_{\le |\lambda| c_i^2 h(X,X)/2 }
 + \underbrace{\lambda (T(\eta)(X))^2}_{\le \lambda c_i^2 h(X,X)} <0
$$
since, in our conventions, causal vectors never vanish.
We can choose $|\lambda_i|$ sufficiently small so that
$$
 d_{B_i}(g(\eta_i(\lambda_i),\lambda_i), g)
  \le \frac \epsilon{2^i}
 \;,
$$
where
\bel{13III11.6}
 d_{B_i}(g_1, g_2):= \sup_{0\ne X,Y \in T_p M, \, p\in B_i(r_i)}
 \frac{|g_1(X,Y)-g_2(X,Y)|}{|X|_h |Y|_h}
 \;.
\ee
Is is now easy to check that, for all $\epsilon>0$ sufficiently small, the smooth tensor field
$$
 \hg :=\sum_i \chi_i g(\eta_i(\lambda_i),\lambda_i)
$$
has Lorentzian signature, has light-cones wider than $g$, and satisfies
\eq{13III11.4}, as desired.

The metric $\chg $ is obtained by choosing $\lambda$ positive in
the construction above. The details are left to the reader.
\qed

\bigskip

Having established the existence of metrics
with the properties spelled-out above, we recall and introduce some definitions:

\begin{Definition}
Let $I$ be an interval.
 We will say that a locally Lipschitz path $\gamma \colon I\to\mcM$, with weak-derivative vector
 $\dot \gamma$ (defined almost everywhere), is
 \begin{enumerate}
 \item \emph{locally uniformly timelike\/} (l.u.-timelike, or l.u.t.) if there
exists a smooth Lorentzian
      metric $\check g \prec g$ such that $\check g( \dot
      \gamma,\dot \gamma) < 0 $ almost everywhere;
       \item \emph{timelike\/} if  we have $  g( \dot
           \gamma,\dot \gamma)< 0$ almost everywhere,
       \item \emph{causal\/} if   $  g( \dot
           \gamma,\dot \gamma)\le 0$ with $\dot \gamma \ne 0$ almost everywhere.
 \end{enumerate}
We set
 \bean
  \cI_g^+ (\Omega,\mcM)&:=&\{q\in \mcM: \ \exists \mbox{\ a future directed l.u.\ $g$-timelike curve  }
  \\
   \label{9IV11.1}
  &&
   \mbox{ $\gamma \colon I \to \mcM$ starting in $\Omega$ and ending at $q$}\}
  \;,
\\
 \nonumber
  I_g^+ (\Omega,\mcM)&:=&\{q\in \mcM: \ \exists \mbox{\  a future directed $g$-timelike curve  }
  \\
   \label{9IV11.2}
  &&
   \mbox{ $\gamma \colon I\to\mcM$ starting in $\Omega$ and ending at $q$}\}
  \;,
\\
 \nonumber
  J_g^+ (\Omega,\mcM)&:=&\Omega\cup\{q\in \mcM: \ \exists \mbox{\  a future directed $g$-causal curve  }
  \\
   \label{9IV11.3}
  &&
   \mbox{ $\gamma \colon I\to\mcM$ starting in $\Omega$ and ending at $q$}\}
  \;.
 \eea
The sets $\cI_g^- (\Omega,\mcM), I_g^- (\Omega,\mcM), J_g^- (\Omega,\mcM)$ are defined in an analogous fashion.
When $I$ is compact, an l.u.-timelike curve $\gamma \colon I \to \mcM$ will be said to be \emph{uniformly timelike}.
\end{Definition}

We will write $\cI^{\pm} (\Omega )$, etc., when the metric and the manifold $\mcM$ are clear from the context. Furthermore $\cI^{\pm}(p):= \cI^{\pm}(\{p\})$, etc.
Clearly
\bel{28IX11.1}
 \cI^+ (\Omega)\subset I^+(\Omega)\subset J^+(\Omega)
 \;.
\ee
(We will, henceforth, often state results for the sets $\cI^+(\Omega)$, etc, on the understanding that analogous statements hold for the corresponding past sets with appropriate modifications.)

We will often use the obvious fact that $\gamma$ is
l.u.-timelike if and only if there exists a smooth Lorentzian
metric $\check g \prec g$ such that $\dot \gamma$ is $\check
g$-\emph{causal\/} almost everywhere.

Smooth curves $\gamma$ in Minkowski space-time with tangent vector $\dot \gamma$ timelike everywhere except at
exactly one point where $\dot \gamma$ is null provide examples of timelike curves which are not l.u.t.

The adjective ``locally" in our definition of l.u.t. curve is motivated by the fact that the light-cones of the metric $\check g$ could be approaching very fast those of $g$ when one recedes to infinity. One could attempt to introduce a notion of ``uniformly timelike" by requiring that the light-cones of $g$ stay a fixed distance apart from the light-cones of $\check g$, where the distance is defined by \eq{13III11.4}. However, there is no natural way of doing this as long as the auxiliary Riemannian metric $h$ is arbitrary.

One immediately finds:

\begin{Proposition}
 \label{P9IV11.1}
 $\cI^+(\Omega)$ is open.
\end{Proposition}

\proof
We have
$$
 \cI^+_g(\Omega)=\cup_{\chg\prec g} I^+_{\chg}(\Omega)
 \;,
$$
with each $I^+_{\chg}(\Omega)$ open by standard results on smooth metrics.
\qed

\medskip

For smooth metrics $\cI^+$ and $I^+$ coincide. We do not know whether or not
this holds for continuous metrics.

 \medskip

The following fact turns out to be useful:

\begin{Proposition}
 \label{P1IV11.1}
 A path $\gamma$ is causal for $g$ if and only if $\gamma$ is causal for every smooth metric $\hg \succ g$.
\end{Proposition}

\proof
Let $\gamma \colon I\to \mcM$ be causal for every smooth metric
$\hg$ satisfying $\hg \succ g$. Suppose that
$\gamma$ is not $g$-causal, then there exists a non-zero
measure set $Z\subset I$ such that the weak derivative $\dot
\gamma$ of $\gamma$ is $g$-spacelike for all $p \in Z$. Let
$W\subset I$ be the set of points at which $\gamma$ is
classically differentiable. Then, by Rademacher's
theorem~\cite{EvansGariepy}, $W$ has full measure in $I$.
Therefore, the set $Z\cap W$ has the same measure as $Z$. In
particular, $Z\cap W$ is not empty.

Given $s_0\in Z\cap W$, there exists a smooth metric $\hg \succ
g$ such that $\dot \gamma (s_0)$ is spacelike for $\hat g$.
Since $\hat g$ is smooth, there exists a normal neighbourhood,
$U$, of $\gamma(s_0)$ such that, for any $q \in U$, there
exists a unique affinely-parametrised $\hg$-geodesic $\hat{\gamma}_q
\colon [0, 1] \to U$ from $\gamma(s_0)=\hat \gamma_q(0)$ to the point $q=\hat \gamma_q(1)$. As
in~\cite[Proposition~2.2.3]{ChCausalityv1}, we define the
function $\hat \sigma \colon U \to \mathbb{R}$ by
$\hat{\sigma}(q) := \hat{g}_{\gamma(s_0)}\left(
\frac{d \hat{\gamma}{}_q}{ds}(0), \frac{d \hat{\gamma}{}_q}{ds}(0) \right)$. Taylor expanding $\hat \sigma$ and $\gamma$ we have
\beaa
 \lefteqn{
\gamma(s) = \gamma(s_0) + \dot \gamma (s_0)(s-s_0) + o(s-s_0)
 }
 &&
\\
 && \ \Longrightarrow
\
 \hat \sigma (\gamma(s)) =\hg(\dot \gamma(s_0), \dot \gamma(s_0)) (s-s_0)^2 + o((s-s_0)^2)
 \;.
\eeaa
Since $\hg(\dot \gamma(s_0), \dot \gamma(s_0)) >0$ the right-hand side is positive for $s$ sufficiently close to $s_0$, which contradicts the fact that $\gamma$ is causal for the metric $\hg$. We conclude that $Z\cap W$ is empty, hence $\dot \gamma$ is causal almost everywhere.

The reverse implication is trivial.
\qed

\medskip

As a corollary we obtain one of the key tools of causality theory:

\begin{Theorem}
 \label{T1IV11.1}
 Let $\gamma_n$ be a sequence of causal curves accumulating at $p$.
Then there exists a causal curve $\gamma$ through $p$ which is an accumulation curve of the $\gamma_n$'s.
\end{Theorem}

Equivalently, there exists a causal curve $\gamma$ and a subsequence
$\gamma_{n_i}$ which converges to $\gamma$ uniformly on compact sets.

 \medskip

\proof
Let $\hg\succ g$ be smooth, then the $\gamma_n$'s are causal
for $\hg$, and by standard results for smooth metrics~\cite[Theorem~2.6.7]{ChCausalityv1}
there exists a $\hg$-causal curve $\gamma$ through $p$ which is an
accumulation curve of the $\gamma_n$'s. Note that the notion of
accumulation curve is independent of the metric, so that the
same curve $\gamma$ is a causal accumulation curve of the
$\gamma_n$'s for any smooth metric with cones larger than $g$.
The result then follows by Proposition~\ref{P1IV11.1}.
\qed
\medskip

We finally introduce
\bean
\hspace{-1cm}\cJ_g^+ (\Omega,\mcM)&:=&\Omega\cup\{q\in \mcM: \ \exists \mbox{\  a  curve  $\gamma \colon I\to\mcM$ starting in $\Omega$ and }
  \\
   \nonumber
  &&
   \mbox{ending at $q$ which is an accumulation curve of a sequence}
   \\
   \label{9IV11.3x}
  &&
   \mbox{of l.u.t. future directed curves starting in $\Omega$}\}
  \;.
 \eea
Curves $\gamma$ as in \eq{9IV11.3x} will be called future directed \emph{$\cJ_g$-causal, or
$\cJ$-causal, curves starting in $\Omega$}. More generally, curves which are
accumulation curves of l.u.t. curves will be called \emph{$\cJ$-causal}.

\medskip

Theorem~\ref{T1IV11.1} and standard considerations show that

\begin{Proposition}
 \label{P29IX11.1}
  \begin{enumerate}
   \item
$\cJ$-causal curves are causal.
 \item
Let $\gamma_n$ be a sequence of $h$-parameterized $\cJ$-causal curves
accumulating at a point $p$. Then there is a $\cJ$-causal curve $\gamma$ through
$p$ which is an accumulation curve of the $\gamma_n$'s.
  \end{enumerate}
\qed
\end{Proposition}

We point out the inclusions
\bel{28IX11.2}
 \cI^+ (\Omega)\subset I^+(\Omega)\;,
  \quad
 \cJ^+ (\Omega)\subset \overline \cI^+(\Omega)\;,
  \quad  \cI^+ (\Omega)\subset \cJ^+(\Omega)\subset J^+(\Omega)
 \;,
\ee
and note that an inclusion relation between $I^+(\Omega)$ and $ \cJ^+(\Omega)$ is not clear.

The following notion is a convenient replacement for the notion of elementary region of \cite[Definition~2.2.7]{ChCausalityv1}:

\begin{Definition}
 \label{D9IV11.1}
An open conditionally compact set $\mcU =I\times \mcV$ will be called a \emph{cylindrical neighborhood\/} of a point $p$ if $\overline{\mcU}$ is contained in the domain of a single coordinate system in which $g_{\mu\nu}=\eta_{\mu\nu}$ at $p$, and in which the coordinate slopes of the light-cones of $g$, when graphed over $\mcV$, are bounded by $1/2$ from below and by $2$ from above.%
\footnote{By ``coordinate slope" we mean the ratio of the $t$-component of a null vector with the Euclidean length of its space-components in the coordinate basis.}
We will further assume that $\nabla t$ is $g$-timelike,
where $t$ is the coordinate along the $I$ factor of $\mcU$.
\end{Definition}

We also introduce:

\begin{Definition}
 \label{D25IV11.1}
  \begin{enumerate}
   \item
A path $\gamma \colon I\to\mcM$ will be called a \emph{limit-geodesic\/} if there exists a sequence of smooth metrics $g_n$ converging locally uniformly to $g$ and a sequence of $g_n$-geodesics $\gamma_n \colon I\to \mcM$ such that the $\gamma_n$'s converge locally uniformly to $\gamma$.

 \item A causal path $\gamma$ through $p$ will be called \emph{approximable\/} if there exists a sequence of metrics $\chg_n\prec g$ converging locally uniformly to $g$ and a sequence of $\chg_n$-causal paths through $p$ which converge locally uniformly to $\gamma$; it will be said to be \emph{non-approximable\/} otherwise.
  \end{enumerate}
\end{Definition}

We have the following replacement for \cite[Proposition~2.4.5]{ChCausalityv1}:

\begin{Proposition}
 \label{P10IV11.1}
Let $g$ be a continuous metric on $\mcM$. Then:
 \begin{enumerate}
\item
 \label{P10IV11.1p1}
Every $p\in \mcM$ has a cylindrical neighborhood $\mcU=I\times \mcV$.
\item
 \label{P10IV11.1p2}
$\dotJ^+_g(p,\mcU)$ is a uniformly Lipschitz graph over $\mcV$, and $q\in \mcU$
lies \underline{above or on} the graph of $\dotJ^+_g(p,\mcU)$ if and only if
$q\in J^+_g(p,\mcU)$. In particular
$$
 \dotJ^+_g(p,\mcU)\subset J^+_g(p,\mcU)
 \;;
$$
equivalently, $ J^+_g(p,\mcU)$ is closed in $\mcU$.
  \item
 \label{P10IV11.1p2b}
 $q\in \mcU$ lies \underline{above} the graph of $\dotJ^+_g(p,\mcU)$ if and only if $q\in \cI^+_g(J^+_g(p,\mcU),\mcU)$.
\item
 \label{P10IV11.1p2c}
 $\dotcI ^+_g(p,\mcU)\subset J^+_g(p,\mcU)$.
\item
 \label{P10IV11.1p2d}
 $\dotcI ^+_g(p,\mcU)$ is a uniformly Lipschitz graph over $\mcV$, and $q\in \mcU$ lies above the graph of $\dotcI ^+_g(p,\mcU)$ if and only if $q\in \cI^+_g(p,\mcU)$.
      \item
 \label{P10IV11.1p5}
 If  $ {\dcI}^+_g(p,\mcU) \not= \dotJ^+_g(p,\mcU)$, then the \emph{future bubble set of $p$,}
$$
 \mcB_p^+:=\cI^-_g(\dotcI ^+_g(p,\mcU),\mcU)\cap \cI^+_g( \dotJ^+_g(p,\mcU),\mcU)
$$
is open non-empty.

\item
 \label{P10IV11.1p5d}
Making $\mcU$ smaller if necessary, for every point $q$ in $J^+_g(p,\mcU)$
there exists a \underline{causal limit-geodesic}
from $p$ to $q$.

\item
 \label{P25V11.1}
 For every point $q$ in
\bel{21V11.5}
I^-_g({\dcI}{}^+_g(p,\mcU),\mcU) \cap J^+_g(p,\mcU)\;,
\ee
every {causal} curve from $p$ to $q$ is \underline{non-approximable}.
 \end{enumerate}
\end{Proposition}

\proof
Let $x^\mu$ be any coordinate system near $p$ with $\partial_0$
timelike. By a Gram--Schmidt orthogonalisation starting from
the basis $\partial_\mu$ we can find a linear map so that
$A^\mu {}_{\nu} \partial_\mu$ forms an $ON$-basis at $p$. An
associated linear coordinate transformation leads to
coordinates in which $g_{\mu\nu}$ takes the Minkowskian form at
$p$; in particular, the light-cones at $p$ have
coordinate-slopes equal to one.
By continuity, the light cones will have slopes between
one-half and two in a sufficiently small coordinate
neighborhood of $p$, which can be taken of the form $I\times
\mcV$. This proves point~(\ref{P10IV11.1p1}).

It is convenient to assume that $x^\mu(p)=0$, which can always be achieved by a translation.

To continue, let $\chg_n$ be a sequence of smooth metrics on
$\mcU$ uniformly converging to $g$ such that $\chg_n \prec
\chg_{n+1} \prec g$. Similarly, let $\hg_n$ be a sequence of
smooth metrics on $\mcU$ uniformly converging to $g$ such that
$\hg_n \succ \hg_{n+1} \succ g$.
Every point in $\mcU=I\times \mcV$ can be written as $(x^0,
\vec x)$. In particular, we can write $\gamma (s) =
(\gamma^0(s),\vec \gamma (s))$. Given $\vec x \in \mcV$, we set
$$
 f_g (\vec x) = \inf_\gamma \gamma^0 (s_{\vec x})
 \;,
$$
where the $\inf$ is taken over all future directed $g$-causal curves with initial point at $p$ for which there exists $s_{\vec x}$ such that $\vec \gamma(s_{\vec x}) = \vec x$. (We define the infimum over an empty set  to be $+\infty$;  there always exists a neighborhood of $p$ on which $f_g <\infty$.)
 It is well-known that if $g$ is smooth, then the infimum is attained on null geodesics passing through $p$ and $(f_g(\vec x) ,\vec x)$, and that $f_g$ is the graphing function of $\dotJ^+_g(p,\mcU)$.

From the definitions it follows immediately that
$$
 f_{\chg_{n}} \ge
 f_{\chg_{n+1 }} \ge
 f_{\hg_{n+1}} \ge
 f_{\hg_{n}}
 \;.
$$
Hence the pointwise limits
\bel{11IV11.1}
 f_-:= \lim_{n\to\infty} f_{\hg_{n}}  \ \mbox{ and } \ f_+:= \lim_{n\to\infty} f_{\chg_{n}}
\ee
exist, with
$$
 f_- \le f_+\;.
$$
For $n$ large enough we can assume that the coordinate slopes of both the $\chg_n$ light-cones and the $\hg_n$ light-cones are bounded by $1/4$ below and by $4$ from above.
The Arzela--Ascoli theorem shows that the limits in \eq{11IV11.1} are uniform, since all the $f_{\chg_{n}} $ and $ f_{\hg_{n}} $ are Lipschitz, with Lipschitz constant less than or equal to four. The limiting functions are thus also Lipschitz, with identical bounds on the Lipschitz constant.

If $\gamma$ is $g$-causal then it is $\hg_n$-timelike, hence
lying above the graph of $f_{\hg_n}$. Thus
$$
 \gamma^0(s_{\vec x})> f _{\hg_n}(\vec x) \ \mbox{for all $g$-causal curves through $p$ and for all $n$}
 \;.
$$
Passing to the limit $n\to\infty$ we obtain
$$
 \gamma^0(s_{\vec x})\ge  f _{-}(\vec x) \ \mbox{for all $g$-causal curves through $p$}
 \;.
$$
We conclude that the image of every $g$-causal curve through $p$ lies above, or on, the graph of $f_-$.

Similarly it follows that $f_+$ is the graphing function for
$\dcI {}^+_g$, the boundary of $\cI^+_g$, and that the image of every
l.u.-timelike $g$-causal curve through $p$ lies above
the graph of $f_+$.
This establishes point~(\ref{P10IV11.1p2d}).

The uniform bounds on the light-cones of all the metrics involved allow us to parameterize inextendible causal curves in $\mcU=I\times \mcV$ by $t\in I$. Hence we can write $\gamma(s)=(s,\vec \gamma(s))$, with $\vec \gamma(s)$ uniformly Lipschitz, with Lipschitz constant less than or equal to four.

We claim that the graph of $f_-$ coincides with $\dotJ_g^+ (p,\mcU)$. To see this, let $\gamma_n(s)=(s,\vec \gamma_n(s))$ be a null geodesic generator of $\dotJ_{\hg_n}^+ (p,\mcU)$ from $p$ to $(f_{\hg_n}(\vec x),\vec x)$ (choose one if there are more than one). By Arzela--Ascoli, from the sequence $\vec \gamma_n$ we can extract a subsequence converging to a Lipschitz curve $\vec \gamma(s)$. Then the curve $\gamma(s)=(s,\vec \gamma (s))$ lies on the graph of $f_-$, and is causal by Proposition~\ref{P1IV11.1}.

Since every point $(t,\vec x)$ lying above the graph of $f_-$ can be connected to the graph
by the future-directed $g$-timelike curve $s\to (s,\vec x)$, point~(\ref{P10IV11.1p2}) follows.
Clearly points above the graph are in $\cI^+_g(J^+_g(p,\mcU),\mcU) $, and point~(\ref{P10IV11.1p2b}) easily follows.

Point~(\ref{P10IV11.1p2c}) should be clear from what has been said so far.

We will show below that there exist metrics for which $f_-\ne f_+$. In this case the set
\bel{12V12.1}
 {\cal B}_p:=\{f_-(x)<t<f_+(x)\}
\ee
is open, non-empty. For any $x$ for which the interval $f_-(x)<t<f_+(x)$ is non-empty the curve
$$
 f_-(x)<t<f_+(x)\ni t\mapsto (t,x)
$$
is future directed timelike. It follows that
$$
 {\cal B}_p
 \subset \mcB_p^+ =
 \cI^-_g(\dotcI ^+_g(p,\mcU),\mcU)\cap \cI^+_g( \dotJ^+_g(p,\mcU),\mcU)
\;.
$$
Equality of ${\cal B}_p$ and $\mcB_p^+$, and hence point~(\ref{P10IV11.1p5}),
follows now from the fact that the graphs of $f_-$ and $f_+$ are separating
hypersurfaces.

To establish point~(\ref{P10IV11.1p5d}), let $\hg_n\succ g$ converge uniformly to $g$ and let $q\in J^+_g(p,\mcU)$. The uniform bound on the slopes of the light-cones within $\mcU$ easily implies that there exists within $\mcU$ a neighborhood $\mcO$ of the origin which is globally hyperbolic for all nearby smooth metrics.
Here, and elsewhere, \emph{global hyperbolicity} is defined as the requirement that $(\mcM,g)$ be stably causal, and that all non-empty sets of the form $J^+(p)\cap J^- (q)$ be compact.
Replacing $\mcU$ be a cylindrical subset of $\mcO$ we conclude, by standard results, that for all $q\in J^+_{g_n}(p,\mcU)$ there exists a $g_n$-causal $g_n$-geodesic $\gamma_n$ from $p$ to $q$,
then the $\gamma_n$'s converge uniformly to a $g$-causal curve $\gamma$ from $p$ to $q$. By definition, $\gamma$ is a limit-geodesic.
(If $q\in \cI_g^+(p,\mcU)$, we can take the $\gamma_n$'s to be   $\check g_n$-timelike geodesics associated with a sequence of metrics $\check g_n\prec g$, but whether or not the resulting limit-geodesic will be timelike is not clear).

To establish point~(\ref{P25V11.1}), suppose that there exists an approximable causal future directed curve $\gamma$ from $p$ to $q$, then $\gamma$ is the limit of curves the images of which lie above the graph of $f_+$. Hence no approximable causal curves from $p$ to points lying under that graph exist, which completes the proof.
\qed

\medskip

We continue with the promised example, where the equality of $f_-$ and $f_+$ fails:

\begin{cocoExa}
 \label{Exa19V11.1}
 {\rm
Let $\lambda>0$ and consider the metric
\bean
 g &= &
 -(du+(1-|u|^\lambda) dx)^2+ dx^2
\\
 & = &
 -du^2 -2(1-|u|^\lambda) du\,dx + |u|^\lambda(2-|u|^\lambda) dx^2
\;.
\eeal{19V11.1}
We have $\det g_{\mu\nu}=-1$, hence $g$ is a Lorentzian metric on $\R^2$ which belongs to $C^{0,\lambda}(\R^2) \cap C^\infty(\R^2\setminus \{u=0\})$ for $\lambda\in (0,1]$,
$C^{1,\lambda-1}(\R^2)\cap C^\infty(\R^2\setminus \{u=0\})$ for $\lambda\in (1,2)$, and to $C^{2}(\R^2)$ for $\lambda\ge2$, in fact smooth for $\lambda\in 2\N$.
Now, in dimension $1+1$, on regions where the metric is differentiable, every continuously differentiable null curve is a causal geodesic.
Consider then the path $x\mapsto \gamma(x)=(u(x),x)$, then $\gamma$ will be null if and only if
\bel{21V11.1}
 u'=\epsilon-(1-|u|^\lambda)\;,
 \qquad
 \epsilon \in \{\pm 1\}
 \;.
\ee
The equation with $\epsilon=1$ has a solution $u(x)\equiv 0$. This is the unique solution for $\lambda\ge 1$, but for $\lambda\in (0,1)$ and $x>0$ we obtain in addition the usual family of bifurcating solutions,
\bel{21V11.3}
 u_{x_0}(x) = \left\{
                \begin{array}{ll}
                 0 , & \hbox{$x_0\ge x \ge 0$;} \\
\left( {(1-\lambda)}{(x-x_0)}\right)^{\frac 1 {1-\lambda}} & \hbox{$ x\ge x_0\ge 0$,}
                \end{array}
              \right.
\ee
see Figure~\ref{F21V11.1}.
\begin{figure}[tbh]
 \begin{center}
 \includegraphics[width=.5\textwidth]{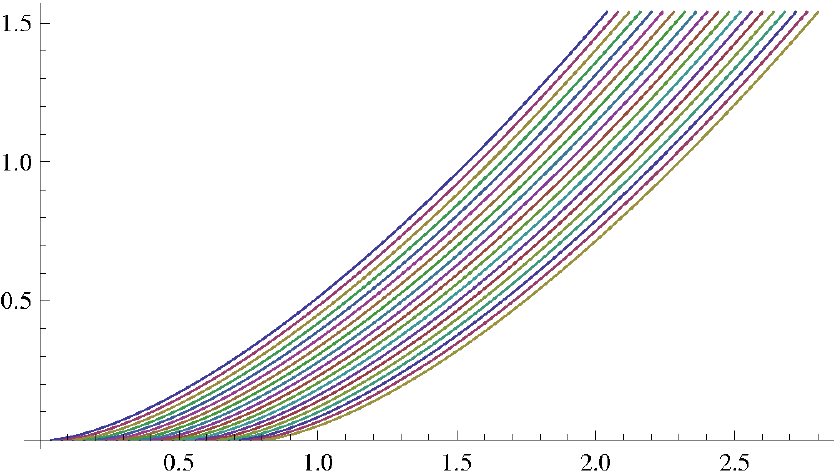}
\caption{\label{F21V11.1} Some right-going null limit-geodesics through $(0,0)$, $\lambda=1/3$ (the coordinate $x$ runs along the horizontal axis). The \emph{causal bubble\/} lies between the $\{u=0\}$ axis and the graph of the ``first bifurcating geodesic" $u_0(x)=(2x/3)^{3/2}$.}
\end{center}
\end{figure}

Consider the functions $f_\pm$ corresponding to the light-cone of the origin $x=u=0$. We claim that, for $x\ge 0$,
\bel{21V11.2}
 f_-(x)\equiv 0\;,
 \qquad
 f_+(x) = u_0(x)
 \;.
\ee
(Negative $x$ can be handled using the isometry $(x,u)\mapsto (-x,-u)$, which leaves both the metric and \eq{21V11.1} invariant).
To establish \eq{21V11.2}, consider any point $p_1=(u_1,x_1)$ with $x_1>0$ and $0\le u_1 \le u_0(x_1)$. Then the unique right-going future directed null $g$-geodesic through $p_1$ belongs to the family \eq{21V11.3}, where $x_0(p_1)>0$ can be calculated from the equation
$$
 u_1 =\left( {(1-\lambda)}{(x_1-x_0(p_1))}\right)^{\frac 1 {1-\lambda}}
 \;.
$$
Let $\check g$ be any smooth metric such that $\check g \prec g$. Then the boundary of the $\check g$-causal past of $p_1$ is a graph of a function which lies below the graph of $u_{x_0(p_1)}$, hence meets the axis $\{u=0\}$ at points with $x$ coordinate {strictly larger} than $x_0(p_1)$. Hence any past directed $\check g$-timelike curve through $p_1$ will also meet $\{u=0\}$ at a point with $x$ coordinate {strictly larger} than $x_0(p_1)$. Since $\{u=0\}$ is $\check g$-spacelike, no $\check g$-timelike curve from this point will ever reach $p_0:=(0,0)$.
So, for $x>0$ the set $\cI(p_0)$ lies above the graph of $x\mapsto u_0(x)$. Note also that every point between the $\{u=0\}$-axis and the graph of $u_0$ lies on the image of  a $g$-causal curve through $p_0$.
The result follows now from Proposition~\ref{P10IV11.1}.

The reader should note that any differentiable curve through the origin entering the bubble region has to have a null tangent vector at the origin. This makes it clear that no differentiable curve through $p_0$ with tangent timelike everywhere enters the bubble region.
In fact, an argument identical to the one in the last paragraph shows that there are no \emph{timelike curves\/} from $p_0$ to points in the bubble region, regardless of differentiability, leading to
$$
 \cI^+(p_0)=I^+(p_0)
  \;.
$$
In particular, the usual Push-up Lemma, which asserts that any causal curve can be deformed slightly to the future to become timelike (compare \cite[Corollary~2.4.16]{ChCausalityv1}),
is \emph{wrong\/} for causal curves from $p_0$ to the bubble region.

We end this example by noting that the scalar curvature $R$ of the metric \eq{19V11.1} equals
\bel{25X11.}
  R = 2 \lambda |u|^{\lambda -2}\left( \lambda - 1 - (2\lambda -1) |u|^\lambda\right)
  \;,
\ee
which is unbounded from below near $\{u=0\}$ for $\lambda\in(0,1)$.
}
\end{cocoExa}

\begin{cocoExa}
 \label{Exa30X11.1}
{\rm
Any function $f(u,x)$ such that the equation $u'=f(u,x)$ has non-unique solutions leads to a bubbling
metric $g$ by setting
$$
 g = -(du +(1-f(u,x))dx)^2 + dx^2
 \;.
$$
An example where the resulting metric is continuous everywhere and smooth except at the origin is thus provided by taking
$f= 4x^3 u/(x^4+u^2)$, with non-unique solutions through the origin
$u=\pm(C^2-\sqrt{x^4+C^4})$, $C\in \R$. (Continuity of $f$ at the origin follows from $2x^2 |u| \le x^4 + u^2$,
whence $|f|\le 2 |x| \le 2\sqrt{x^2+u^2}$.) See Figure~\ref{F30X11.1}.
\begin{figure}[ht]
 \begin{center}
 \includegraphics[width=.5\textwidth]{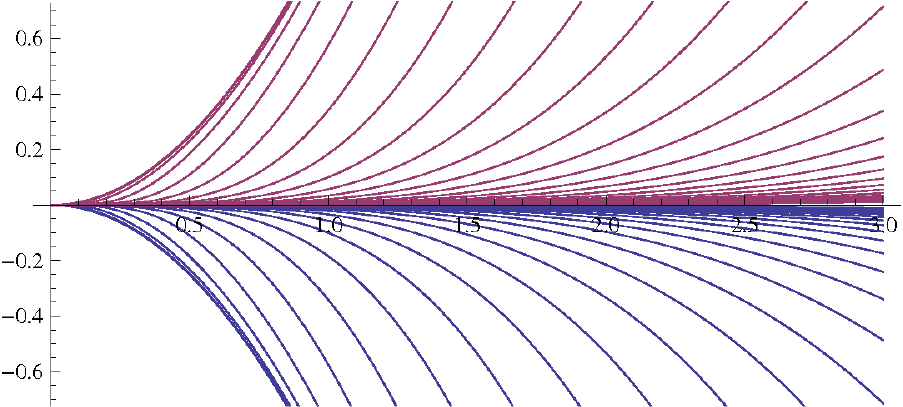}%
 \includegraphics[width=.4\textwidth]{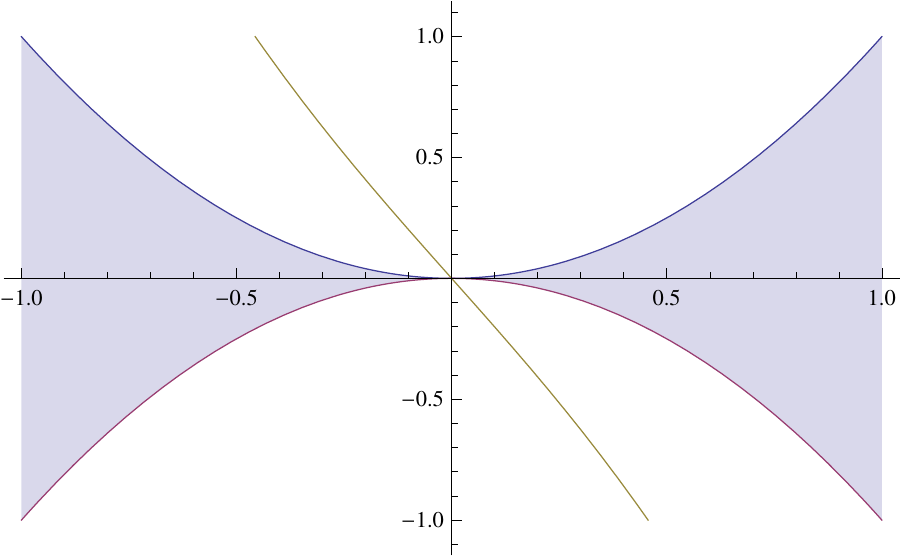}
\caption{\label{F30X11.1} In Example~\ref{Exa30X11.1}, the images of the future
right-going null curves issuing from the origin fill the region $|u|\le x^2$, $x>0$. The set covered by both the left- and right-going null geodesics through the origin is the union of the shaded region and of the negative-slope curve through the origin in the right figure.}
\end{center}
\end{figure}
}
\end{cocoExa}

\begin{cocoExa}
 \label{Exa25V11.3}
{\rm
Denote by ${}^2 g$ the metric from Example~\ref{Exa19V11.1}, and let $\delta$ denote the Euclidean metric on $\R^{n-1}$. Set
$$
 g = {}^2g+ \delta
 \;.
$$
It follows immediately from what has been said in Example~\ref{Exa19V11.1} that
$$
 \dotI ^+(\{(0,0)\}\times \R^{n-1}) = \dcI^+(\{(0,0)\}\times \R^{n-1}) \ne \dotJ^+(\{(0,0)\}\times \R^{n-1})
  \;,
$$
and that
$$
  J^+(\{(0,0)\}\times \R^{n-1}) \setminus   I^+(\{(0,0)\}\times \R^{n-1})
$$
has non-empty interior.
}
\end{cocoExa}

On the other hand, there are cases of interest where we can establish equality of $f_-$ and $f_+$.

\begin{cocoExa}
 \label{Exa15V11.3}
{\rm
Suppose that $g$ is $C^{1,1}$. For any sequence of $C^{1,1}$ metrics $g_n$ converging in $C^{1}$ to $g$, causal $g_n$-geodesics converge to causal $g$-geodesics. Since the latter are uniquely determined by the initial data, we conclude that $f_+=f_-$.
}
\end{cocoExa}

In fact, we can prove more.%
\footnote{Closely related results have been established by A.~Fahti and A.~Siconolfi in the general context
of~\cite{FathiSiconolfiTime} (A.~Fathi, unpublished, private communication).}

\begin{Lemma}
 \label{LT10XI11.1}
Let $g$ be a locally Lipschitz Lorentzian metric on $\mcM$. Let $\gamma \colon [0, 1] \to \mcM$ be a
 future-directed Lipschitz causal
curve,
 contained in a single (open) convex coordinate chart $\mcU\subset\mcM$, with $g(\partial_0,\partial_0)\le -\varepsilon < 0$ for some $\varepsilon>0$.
 Then, for any $q\in \cI^+(\gamma(1))\cap \mcU $ such that
 $$x^0(q)> x^0 (\gamma(1))\;,\qquad x^i(q)= x^i (\gamma(1))
 \;,
 $$
 there exists a (Lipschitz) locally uniformly timelike curve from $p := \gamma(0)$ to $q$.
\end{Lemma}

\proof
Near $\gamma([0,1])$ we have coordinates $\{ x^{\alpha} \}_{\alpha = 0}^n$ in which the components of the metric,
$g_{\alpha\beta}(x)$, are Lipschitz functions of $x=(x^{\mu})$. Let $T = \partial_0\equiv \partial_t$, thus $T$ is timelike by hypothesis.

In this coordinate system, the curve $\gamma$ takes the form $s \mapsto \gamma^{\alpha}(s)$. Given a function $f \colon [0, 1] \to \R$, we define the new curve $\Gamma \colon [0, 1] \to M$ by $\Gamma^{\alpha}(s) = \gamma^{\alpha}(s) + \eps f(s) T^{\alpha}$. Our aim is to find an appropriate function $f$ such that, for sufficiently small $\eps > 0$, the curve $\Gamma$ has the required properties.
As such, we impose that $f \ge 0$ and, since we wish to construct a curve emanating from the point $p$ we require
\[
f(0) = 0
\;.
\]
Now,
\bean
g_{\Gamma(s)}(\dot{\Gamma}(s), \dot{\Gamma}(s)) &=&
\left( g_{\alpha\beta}(\Gamma(s)) - g_{\alpha\beta}(\gamma(s)) \right)
\frac{d\Gamma^{\alpha}}{ds} \frac{d\Gamma^{\beta}}{ds}
\nonumber
\\
& &+ g_{\gamma(s)} \left( \dot{\gamma}(s) + \eps \dot{f}(s) \partial_t, \dot{\gamma}(s) + \eps \dot{f}(s)
\partial_t \right)
 \;.
\nonumber
\eea
We first note that
\bean
g_{\gamma}(\dot{\gamma} + \eps \dot{f} \partial_t, \dot{\gamma} + \eps \dot{f} \partial_t )
&=
g_{\gamma}(\dot{\gamma}, \dot{\gamma})
+ 2 \eps \dot{f}(s) g_{\gamma} ( \dot{\gamma}, \partial_t)
+ \eps^2 \dot{f}^2 g_{\gamma}(\partial_t, \partial_t)
\\
&\le
g_{\gamma}(\dot{\gamma}, \dot{\gamma})
+ 2 \eps \dot{f}(s) g_{\gamma} ( \dot{\gamma}, \partial_t)
 \;,
\nonumber
\eea
since $g_{\gamma}(\partial_t, \partial_t) < 0$. Moreover, letting $\Lambda$ denote the Lipschitz constant of the metric and $\left| \cdot \right|$ the Euclidean norm, we have
\bean
\left( g_{\alpha\beta}(\Gamma(s)) - g_{\alpha\beta}(\gamma(s)) \right) \dot{\Gamma^{\alpha}} \dot{\Gamma^{\beta}}
&\le&
\left| g(\Gamma(s)) - g(\gamma(s)) \right| \left| \dot{\Gamma} \right|^2
\nonumber
\\
&\le&
\Lambda \left| \Gamma(s) - \gamma(s) \right| \left| \dot{\Gamma} \right|^2
\nonumber
\\
&\le&
\eps \Lambda f \left| \partial_t \right| \left| \dot{\gamma} + \eps \dot{f} \partial_t \right|^2
\nonumber
\\
&\le&
2 \eps \Lambda f \left| \partial_t \right| \left[ \left| \dot{\gamma} \right|^2 + \eps^2 \dot{f}^2 \left| \partial_t \right|^2 \right] \;,
\nonumber
\eea
where we have used the fact that $f \ge 0$ in the penultimate inequality.
We therefore have, for $0\le \epsilon f\le 1$, $0\le \epsilon \le 1$,
with $\epsilon f$ small enough in any case so that the image of the curve lies within the domain of definition of the coordinates,
\bean
g_{\Gamma}(\dot{\Gamma}, \dot{\Gamma})
&\le&
g_{\gamma} \left( \dot{\gamma}, \dot{\gamma} \right) - \eps A \dot{f} + \eps B f + C \eps^3 f \dot{f}^2,
\eea
where we define the functions
\bean
A(s) &:=& 2 \left| g_{\gamma(s)} ( \dot{\gamma}(s), \partial_t) \right|,
\\
B(s) &:=& 2 \Lambda \left| \partial_t \right| \left| \dot{\gamma}(s) \right|^2,
\nonumber
\\
C(s) &:=& 2 \Lambda \left| \partial_t \right|^3.
\nonumber
\eea
Note that $A, B, C$ are positive functions and are uniformly bounded along the (compact) curve $\gamma$.
Assuming that the curve $\gamma$ is causal we have
\[
g_{\Gamma}(\dot{\Gamma}, \dot{\Gamma}) \le - \eps A \dot{f} + \eps B f + C \eps^3 f \dot{f}^2,
\]
and we wish to show that we can find a function $f$ and $\eps > 0$ such that the right-hand-side is strictly negative. Let
\[
f(s) := \int_0^s \frac{1}{A(t)} \exp \left( \int_t^s \frac{B(r)}{A(r)} \, dr \right) \, dt
\]
be the solution of $A(s) \dot{f}(s) - B(s) f(s) = 1$ with $f(0) = 0$. We then deduce that
\[
g_{\Gamma}(\dot{\Gamma}, \dot{\Gamma}) \le - \eps + C \eps^3 f \dot{f}^2.
\]
Since $f$ and $\dot{f}$ are bounded, we may choose $\eps_0 > 0$ such that $\eps_0^2 C(s) f(s) \dot{f}(s)^2 \le \frac{1}{2}$ for all $s \in [0, 1]$. Fixing this value of $\eps$, the corresponding curve $\Gamma$ is timelike wherever differentiable:
\[
g_{\Gamma}(\dot{\Gamma}, \dot{\Gamma}) \le - \frac{1}{2} \eps_0 < 0.
\]

Note that $\dot{\Gamma}$ is uniformly bounded away from the light-cones
of the metric $g$. Given $\alpha > 0$, we define the metric
\[
g_{\alpha} := g + \alpha g(T, \cdot) \otimes g(T, \cdot).
\]
We then have
\[
g_{\alpha}(\dot{\Gamma}, \dot{\Gamma}) = g(\dot{\Gamma}, \dot{\Gamma}) + \alpha \left( g(T, \dot{\gamma}) + \eps_0 f g_{tt} \right)^2
\le - \frac{1}{2} \eps_0 + \alpha \left( g(T, \dot{\gamma}) + \eps_0 f g_{tt} \right)^2 \;.
\]
Since $g(T, \dot{\gamma}) + \eps_0 f g_{tt}$ is bounded along $\gamma$, taking $\alpha > 0$ sufficiently small, we may ensure that, say, $g_{\alpha}(\dot{\Gamma}, \dot{\Gamma}) \le - \frac{1}{4} \eps_0$. Smoothing the metric $g_\alpha$ then gives a smooth metric $\chg \prec g$ such that $\Gamma$ is time-like with respect to $\chg$. Therefore $\Gamma$ is locally uniformly timelike.

Concatenating $\Gamma$ with the timelike curve
$$
 \lambda\mapsto (\lambda,x^i(\gamma(1)))\;, \ \mbox{ where } \  \lambda \in [x^0(\gamma(1)) + \eps f(1), x^0(q)]
 \;,
$$
provides the desired l.u.t.\ curve.
\qed

\medskip

It is convenient to introduce:

\begin{Definition}
\label{D18X11.1}
 A space-time $(\mcM,g)$ will be called \emph{causally plain\/} if for every $p\in\mcM$ there exists a cylindrical neighbourhood $\mcU$ thereof such that
$$
 \dcI^\pm(p,\mcU)= \dotJ^\pm(p,\mcU)
 \;.
$$
$(\mcM,g)$ is said to be \emph{causally bubbling\/} otherwise.
\end{Definition}

As a corollary of Lemma~\ref{LT10XI11.1} we obtain:

\begin{Corollary}
 \label{TC10XI11.1}
 Space-times $(\mcM,g)$ with Lipschitz-continuous metrics are causally plain.
\end{Corollary}

\proof
Lemma~\ref{LT10XI11.1} shows that any curve, from   any $p \in \mcM$, with image in a cylindrical neighbourhood of $p$, and on the graph $t = f_-(x)$, may be perturbed by an amount as small as desired to the future to give a locally uniformly time-like curve. As such, the bubble set of $p$, $\mcB_p^+$, is empty.
\qed

\bigskip

\begin{cocoExa}
 \label{Exa15V11.1}
 {\rm
Let us show that null geodesics  for Lipschitz-continuous metrics are unique in dimension two; this has some interest in its own, but also provides an alternative proof of Corollary~\ref{TC10XI11.1} in that dimension.
Let $g$ be a locally Lipschitzian Lorentzian metric on a two-dimensional manifold $\mcM$,
and let $\mcU$ be a cylindrical neighborhood of $p\in \mcM$.
For any sequence of metrics $g_n$ converging uniformly to $g$ on $\mcU$
let $\{\theta_n^0, \theta_n^1\}$ be an ON-coframe constructed by a
Gram--Schmidt procedure starting from the coframe $\{dt,dx\}$,
where the coordinates $(t, x)$ are chosen such that $g = -dt^2 + dx^2$ at $p$.
Then the sequence $\{ \theta^\mu_n \}_{n\in\N}$ converges uniformly to the
corresponding coframe $\{\theta_\mu\}$ for the metric $g$. We
have $\theta^\mu(\partial_\nu)=\delta^\mu_\nu$ at $p$ hence,
for any $\epsilon>0$ we can find a cylindrical neighborhood
$\mcU_\epsilon\subset \mcE$ such that
\bel{15V11.1}
 |\theta^\mu_n(\partial_\nu)-\delta^\mu_\nu|\le \epsilon
 \;,
\ee
throughout $\mcU_\epsilon$, for all $n$ large enough.

Parameterising the null $g_n$--geodesics $\gamma_n$ through $p$ by $x$, we have $\dot \gamma_n = f_n'\partial_t +
 \partial_x$. The null character of $\dot \gamma_n$ gives
\bel{15V11.2}
 \theta^0_n(\dot\gamma_n) = \pm \theta^1_n(\dot \gamma_n)
 \;.
\ee
Choosing the plus sign (the calculations for the minus sign are similar), one finds
$$
 f'_n = \frac{\theta^1_n(\partial_1) - \theta^0_n(\partial_1)}
{\theta^0_n(\partial_0) - \theta^1_n(\partial_0)}
\;,
$$
and note that the denominator is bounded away from zero for $\epsilon < 1/2$.
It follows that $f'_n$ converges uniformly to $f'$ on, say, $\mcU_{1/8}$, where $f$ is a solution of
$$
 f'  = \frac{\theta^1(\partial_1) - \theta^0(\partial_1)}
{\theta^0(\partial_0) - \theta^1(\partial_0)}
\;.
$$
In view of our assumption that $g$ is Lipschitz-continuous, the solutions of this equation are unique, which proves that $f_+=f_-$, as desired.
}
\end{cocoExa}

\begin{Remark}
{\rm
In Riemannian geometry, Hartman~\cite{HartmanGeodesics} studied $C^{1, \alpha}$ metrics,
with $\alpha < 1$, for which the solutions of the geodesic equations are non-unique.
A straightforward calculation shows that the curvature of Hartman's metrics, as for our metric above,
is unbounded below on the set where the geodesics branch. On the other hand, it is known~\cite{GrovePetersen} that for Riemannian metrics that arise as (Gromov--Hausdorff) limits of metrics with sectional curvature uniformly bounded below, geodesics do not branch. Similarly, it is known that geodesics do not branch for
Alexandrov spaces with curvature bounded below~\cite{BuragoGP}.
}
\end{Remark}

Summarising, we have shown that:

\begin{Theorem}
 \label{TP23V11.1}
\begin{enumerate}
\item For any $\alpha \in (0,1)$ there exist $C^{0,\alpha}$ metrics which are causally bubbling.
\item $C^{0,1}$ metrics are causally plain.
\end{enumerate}
\end{Theorem}

\subsection{Push-up lemmas and their consequences}
 \label{S28IX11.1}

By inspection of \cite{ChCausalityv1}, one finds that the results
there which do
not \emph{explicitly\/} involve geodesics or normal coordinates can be carried
over word-for-word, or with elementary modifications, for \emph{causally plain\/}
space-times using Theorem~\ref{T1IV11.1} and the following four results:

\begin{Proposition}
\label{P20XI11.1}
Let $(\mcM,g)$ be a $C^1$ space-time with a continuous causally plain metric $g$.
Then for all $\Omega\subset \mcM$ we have
  \bel{20XI11.1}
   I^\pm(\Omega)=\cI^\pm(\Omega)
   \;.
   \ee
 \end{Proposition}

\begin{Lemma}[``Push-up Lemma I" (compare \protect{\cite[Lemma~2.4.14]{ChCausalityv1}})]
 \label{CLpushup0}
Let $(\mcM,g)$ be a $C^1$ space-time with a continuous causally plain metric $g$.
For any $\Omega\subset\mcM$ we have
\bel{Ctriinc}I^+(J^+(\Omega))= I^+(\Omega)\;.\ee
\end{Lemma}

\begin{Proposition}[compare \protect{\cite[Corollary~2.4.16]{ChCausalityv1}}]
 \label{CCPushup0}
Let $(\mcM,g)$ be a $C^1$ space-time with a continuous causally plain metric $g$.
Consider a \underline{causal} future directed curve
$\gamma \colon [0,1]\to\mcM$ from $p$ to $q$. If there exist
$s_1,s_2\in[0,1]$, $s_1<s_2$, such that $\gamma|_{[s_1,s_2]}$ is
\underline{timelike}, then for any neighborhood $\mcO$ of
the image of $\gamma$
there exists a \underline{timelike} future directed
curve $\hat\gamma$ from $p$ to $q$ with image contained in that neighborhood.
The curve $\hat \gamma$ can be chosen to be l.u.t. if $\gamma|_{[s_1,s_2]}$ is.
\end{Proposition}

\begin{Lemma}[``Push-up Lemma II" (compare \protect{\cite[Lemma~2.9.10]{ChCausalityv1}})]
 \label{CLpushup}
Let $(\mcM,g)$ be a $C^1$ space-time with a continuous causally plain metric $g$.
Let $\gamma \colon \R^+\to \mcM$ be a past-inextendible past-directed
 \underline{causal}
curve starting at $p$, and let $\mcO$ be a neighborhood of the
image $\gamma(\R^+)$ of $\gamma$. Then for every $r\in I^+({q})\cap\mcO$ there
exists a past-inextendible past-directed
 \underline{timelike}
curve $\hat\gamma \colon [0,\infty)\to \mcM$ starting at $r$ such that
\beal{C27IV11.1}
 &
 \hat \gamma([0,\infty)) \subset
 I^+(\gamma)\cap \mcO\;,
 &
\\
 &
 \forall\ s\in [0,\infty) \qquad I^+(\gamma (s))\cap \hat \gamma(\R^+) \ne
\emptyset\;.
 &
\eeal{C27IV11.2}
\end{Lemma}

We emphasise that all three results are wrong for general continuous metrics in
view of Example~\ref{Exa19V11.1}.

\medskip

To prove the claims, we start with the:

\medskip

{\sc\noindent Proof of Proposition~\ref{P20XI11.1}:}
It suffices to prove the result for $\Omega=\{p\}$, with $\mcM$ replaced by a cylindrical neighborhood $\mcU$ of $p$. Let $f_-$, $f_+$, $\hg_n$, $\chg_n$, etc., be as in the proof of Proposition~\ref{P10IV11.1}. By that proposition, any point above the graph of $f_+$ lies in $\cI^+(p)$, and any point in $I^+(p)\subset J^+ (p)$ lies above or on the graph of $f_-$. By hypothesis we have $f_-=f_+=:f$, and we need to show that no point on the graph of $f$ lies on the image of a future directed timelike curve through $p$.

We start by noting that no curve passing through $p$  \emph{with its image on the graph} of $f$ can be future-directed timelike. Indeed, let $\gamma$ be such a curve, thus $\gamma(s)=(f(\vec x(s)), \vec x(s))$. Consider any point $s_0$ at which $\gamma$ is classically differentiable and at which $T:=\dot \gamma(s_0)$ is timelike. By a linear change of coordinates we can achieve that the metric at $\gamma(s_0)$ takes the standard Minkowskian form. A linear rescaling of the parameterization of $\gamma$ and a linear transformation of the coordinates lead to $T=(1,\vec 0)$. In those coordinates all vectors $\eta=(\eta^0=1,\vec \eta)\in T_{\gamma(s_0)}\mcM $ with $\vec \eta$ having Euclidean norm $|\vec \eta|$ smaller than $1$ are timelike.
Because the metric is continuous, there exists a neighborhood $\mcW$ of $\gamma(s_0)$ so that for all $q\in \mcW$ vectors $X=(X^0=1,\vec X)$ at $q$ with $\vec X$ having Euclidean norm $|\vec X|$ smaller than $1/2$ are timelike.
Hence all points $q\in\mcW$ with coordinates
\bel{21XI11.1}
 q^\mu = \gamma^\mu(s_0)+ \lambda X^\mu \;, \lambda \in (0,1]\;,\quad X^0=1\;,\ |\vec X|< 1/2
\ee
lie to the l.u.t. future of $\gamma(s_0)$. Let us denote by $\Omega_0\subset \mcM$ the set of points of the form \eq{21XI11.1}.

We claim that points in $\Omega_0$ lie strictly above the graph of $f$. Indeed, we can construct a future directed causal curve, say $\gamma_X$, from $p$ to points in $\overline{\Omega_0}$ by following $\gamma$ from $p$ to $\gamma(s_0)$, and then the curve obtained by varying $\lambda$ from $0$ to $1$ in \eq{21XI11.1}. Then $\gamma_X$ is $g$-causal, hence $\hg_n$ timelike for all $n$. It follows that the closure of $\Omega_0$ lies above the graphs of $f_{\hg_n}$ for all $n$, and hence above or on the graph of $f_-$. Thus $\Omega_0$ lies entirely above the graph of $f$.

Now, in the coordinates above we have the Taylor expansion
$$
 \gamma^\mu (s) = \gamma^\mu(s_0)+ T(s-s_0)+o(s-s_0)
 \;.
$$
One easily checks that for all $s-s_0>0$ small enough $\gamma(s)$ belongs to $\Omega_0$, hence lies above the graph of $f$, which is a contradiction.

It remains to exclude the possibility that $\gamma$ enters the strict epigraph of $f$ and returns to the graph.
(Note that $\dot \gamma$ must fail to be differentiable, or timelike if differentiable, at the return point, otherwise the previous argument would apply.)
Consider then a causal curve $\gamma$ from $p$ to some point $q\ne p$ on the graph of $f$. Let $f^-_{q}$ denote the graphing function of the past of $q$. Since $q$ has no causal bubble, any point lying under the graph of $f^-_q$ lies on a l.u.t. past directed curve starting at $q$. Hence $p$ cannot lie under that graph, otherwise $q$ would be in the l.u.t. future of $p$, and would then lie strictly above the graph of $f$. Using again the fact that $(\mcM,g)$ is causally plain, no points lying above the graph of $f^-_q$ are $g$-causally related to $q$. So $\gamma$ must lie on the graph of $f^-(q)$. But the previously given argument applies to $f^-_q$, and shows that no timelike curve can have image on that graph. We conclude that there are no timelike curve from $p$ to points lying on the graph of $f$, and hence $I=\cI$.
\qed

\medskip

We continue by noting that Lemma~\ref{CLpushup0} easily
follows from Proposition~\ref{CCPushup0}, which we prove now:

  \medskip

{\noindent \sc Proof of Proposition~\ref{CCPushup0}:}
Suppose, first, that $s_2=1$ and $s_1 > 0$.
Using compactness, we can cover
$\gamma([0,s_1])$ by a finite collection $\mcU_i$, $i=1,\cdots,N$, of
cylindrical regions $\mcU_i$, entirely contained in $\mcO$, centred at
$p_i\in\gamma([0,s_1])$, with $p_N=\gamma(0)= p$
and for $i=1,\ldots, N-1$,
$$
 p_{i}\in \mcU_{i}\cap\mcU_{i+1}\;, \quad
 p_{i+1} \in J^-{(p_i)}\;,
 \quad
  p_1=\gamma(s_1)
  \;.
$$
Let $\hat s_1$ satisfying $s_1<\hat s_1 $ be close enough to
$s_1$ so that $\hat p_1:=\gamma(\hat s_1)\in\mcU_1$. Since
$\dotcI ^+_g(p_1,\mcU_1)= \dotJ^+_g(p_1,\mcU_1)$ by hypothesis,
point~\ref{P10IV11.1p2d} of Proposition~\ref{P10IV11.1} implies that
there exists a past directed l.u.t. curve $\gamma_1 \colon [0,1]\to\mcU_1$ from
$\hat p_1$ to
$p_2\in\mcU_1\cap\mcU_2$. For $s$ close enough to $1$ the curve $\gamma_1$
enters
$\mcU_2$, choose $\hat s_2\ne 1$ such that $\gamma_1(\hat s_2)\in \mcU_2$.
Again by point~\ref{P10IV11.1p2d} of
Proposition~\ref{P10IV11.1} there exists a past directed
l.u.t. curve $\gamma_{2} \colon [0,1]\to\mcU_2$
from $\gamma_1(\hat s_2)$ to $p_3$. Repeating this construction iteratively,
one obtains a (finite) sequence of
past-directed
l.u.t. curves $\gamma_i\subset I^+(\gamma)\cap \mcO$ such that the end point
$\gamma_i(\hat s_{i+1})$
of $\gamma_i|_{[0,\hat s_{i+1}]}$
coincides with the starting point of $\gamma_{i+1}$.
Concatenating those curves together gives the desired path.

The case $s_1=0$ follows from the case $s_2=1$ by changing time orientation.

The remaining cases easily follow.
\qed

\medskip

Finally, the proof of Lemma~\ref{CLpushup} is a straightforward variation of
that of Proposition~\ref{CCPushup0}: the finite covering $\mcU_i$ of that proof
can be replaced by a countable one, and the distance between the $p_i$'s and
the points $\hat s_i$ should be chosen small enough to make sure that the
resulting concatenation of curves has infinite $h$-length;
compare the proof of~\cite[Lemma~2.9.10]{ChCausalityv1}.
\qed

\medskip

From what has been said so far, we conclude:

\begin{Theorem}
 \label{T23V11.1}
Those parts of causality theory as described in \cite{ChCausalityv1} which do not
explicitly address normal neighborhoods
or geodesics remain true for causally plain space-times.
\qed
\end{Theorem}

This, and Theorem~\ref{TP23V11.1} lead to:

\begin{Corollary}
 \label{C23V11.1}
Those parts of causality theory as described in \cite{ChCausalityv1} which do not
explicitly address normal neighborhoods
or geodesics remain true for $C^{0,1}$ metrics.
\qed
\end{Corollary}

We note that several statements in~\cite{ChCausalityv1} concerning geodesics remain true for $C^{1,1}$ metrics; it is conceivable that all of them remain true, but justifications would be needed.

We continue with some more examples:

\begin{cocoExa}
 \label{Exa15V11.2}
{\rm
Suppose that $(\mcM,g)$ is an $n+1$-dimensional Lorentzian manifold with a \emph{locally Lipschitzian\/} metric $g$.
Suppose moreover that $g$ is invariant under the action of a \emph{compact\/} group $G$,
 the orbits of which are \emph{spacelike\/} and have \emph{co-dimension one}.
The metric is then smooth in the group direction, and one could expect that the question of uniqueness of geodesics would be similar to that in dimension two.
To address this, let $X_a$, $a=1,\ldots,n$ be $n$ linearly independent Killing vectors at $p$. Set $h_{ab}=g(X_a,X_b)$, then $h$ is invertible near $p$, let $h^{ab}$ denote the inverse matrix. One defines the \emph{orbit space-metric\/} $q$ by the formula
$$
 q(X,Y) := g(X,Y) - h^{ab}g(X_a,X) g(X_b,Y)
 \;.
$$
Then $q$ is a metric of signature $(-)$ on the manifold obtained by dividing $\mcM$ by the action of the isometry group.
Let $t$ be a Gauss
 coordinate parameterising the orbits of the isometry group near a point $p\in \mcM$, and let $\gamma$ be a null geodesic through $p$, set $c_a = g(X_a,\dot \gamma)$. Then
\bel{22V11.10}
  -(\dot\gamma^t)^2=q(\dot \gamma,\dot \gamma) = - h^{ab}g(X_a,\dot \gamma) g(X_b,\dot\gamma) = - h^{ab}c_a c_b
 \;;
\ee
equivalently
\bel{22V11.11}
\frac{dt}{ds} =\epsilon \sqrt{  { h^{ab}c_a c_b} }
 \;,
\qquad
\epsilon \in \{\pm 1\}
 \;.
\ee

Since the orbits of $G$ are spacelike we have $h^{ab}c_ac_b>0$, hence $\sqrt{h^{ab}c_ac_b}$ is a Lipschitz-continuous function, and so for any set of constants of motion $(c_a)$ the solutions of \eq{22V11.10},
are unique.

Remark that this fails
when the orbits of the group are not spacelike, as then the square-root fails to be Lipschitz at points where $h^{ab}c_ac_b=0$. An example where this occurs, and where the solutions of \eq{22V11.11} are not unique, has been  presented in Example~\ref{Exa19V11.1} where $\partial_x$ is a Killing vector field; compare Example~\ref{Exa22V11.10} below, where $\partial_t$ and $\partial_y$ are Killing vectors.
}
\end{cocoExa}

\begin{cocoExa}
\label{Exa22V11.10}
{\rm
We have seen that for $C^{0,1}$  the interface between the timelike future and the rest of the world has codimension one. This raises naturally the question, whether \emph{bifurcating null} geodesics can exist on this hypersurface. In this example, we exhibit a family of causally plain metrics of $C^{1,\gamma}$ differentiability class for which part of the light-cone is a hypersurface covered by bifurcating null geodesics all emanating from a point $p$ in the same
direction.

Let $\mcM = \R^3$ with coordinates $(t, x, y)$. Let $\gamma$ be a constant in $(0, 1)$ and $\epsilon$ is a small real number. Consider the following one-parameter family of three-dimensional metrics
\[
 g(\epsilon) = -e^{\alpha_\epsilon(x)}dt^2 + dx^2 + e^{\beta_\epsilon(x)}dy^2
 \;,
\]
with
\[
 e^{-\alpha_\epsilon(x)}  = 1 + \epsilon + \frac{1} 2 (\epsilon^2 + x^2)^{\frac{1+\gamma}2}
 \;,
 \quad
 e^{-\beta_\epsilon(x)}  = 1 - \frac{1} 2 (\epsilon^2 + x^2)^{\frac{1+\gamma}2}
 \;.
\]
%
The interest of these metrics with $\eps = 0$ stems from the fact that they possess bifurcating null geodesics~\cite[Appendix~F]{SCC}. In this case we have
%
\[
 e^{-\alpha_0(x)}  = 1 + \frac{1} 2  |x|^{ {1+\gamma} }
 \;,
 \quad
 e^{-\beta_0(x)}  = 1 - \frac{1} 2  |x|^{ {1+\gamma} }
 \;,
\]
%
so the metric is of $C^{1,\gamma}$ differentiability class. Note that the metric is causally plain by Theorem~\ref{TP23V11.1}. For $\epsilon\ne 0$ the metrics are smooth approximations to $g(0)$, with light cones slightly larger or smaller than those of $g(0)$ depending upon the sign of $\epsilon$.

The null geodesics of the metric $g(\epsilon)$ are obtained by solving the equations
\[
 \dot t = C_t \, e^{-\alpha_\epsilon (x)}
 \;,
 \quad
 \dot y = C_y \, e^{-\beta_\epsilon (x)}
 \;,
\quad
 \dot x^2 = C_t^2 e^{-\alpha_\epsilon(x)} - C_y ^2 e^{-\beta_\epsilon(x)}
 \;,
\]
for some constants $C_t$ and $C_y $. If $C_t^2=C_y^2 =1$ and $\epsilon =0$ we obtain
\[
 \dot x^2 =   |x|^{ {1+\gamma} }
 \;,
\]
and the initial value $x(0)=0$ gives both the solution
$$x(s)\equiv 0\;,\quad t=t_0\pm s\;, \quad y = y_0\pm s
 \;,
$$
and, restricting ourselves to the future light-cone,
$$
 x(s) = \pm \left(\frac{( 1-\gamma )s}{2}  \right)^{\frac{2}{1-\gamma}}
 \;,
$$
$$
 t(s) = t_0  +   s + \frac 1 {(3+\gamma)}\left(\frac{( 1-\gamma )s}{2}  \right)^{\frac{3+\gamma}{1-\gamma}}
 \;,
\quad
 y(s) = y_0 \pm \left( s - \frac 1 {(3+\gamma)}\left(\frac{( 1-\gamma )s}{2}  \right)^{\frac{3+\gamma}{1-\gamma}} \right)
 \;,
$$
see the first two figures in Figure~\ref{F12IV11.1}.%
\footnote{We are grateful to J.M.~Martin-Garcia for help with producing the figures.}
All four combinations of $\pm 1$ in $x$ and $y$ are
possible. Since $3+\gamma/(1-\gamma)>3$, the dominant term both
in $t$ and $y$ is $s$ for small $s$, and note that $t(s)> t_0 +
s $ for $s \neq 0$.

The above solution will be called the \emph{extremal bifurcating geodesic}. Further bifurcating geodesics passing through the origin are obtained by moving the initial point of the bifurcating geodesic along the null ray $\{x=0,t=y\}$, see the third and fourth figures of Figure~\ref{F12IV11.1}.
\begin{figure}[ht]
 \begin{center}
 \includegraphics[width=.15\textwidth]{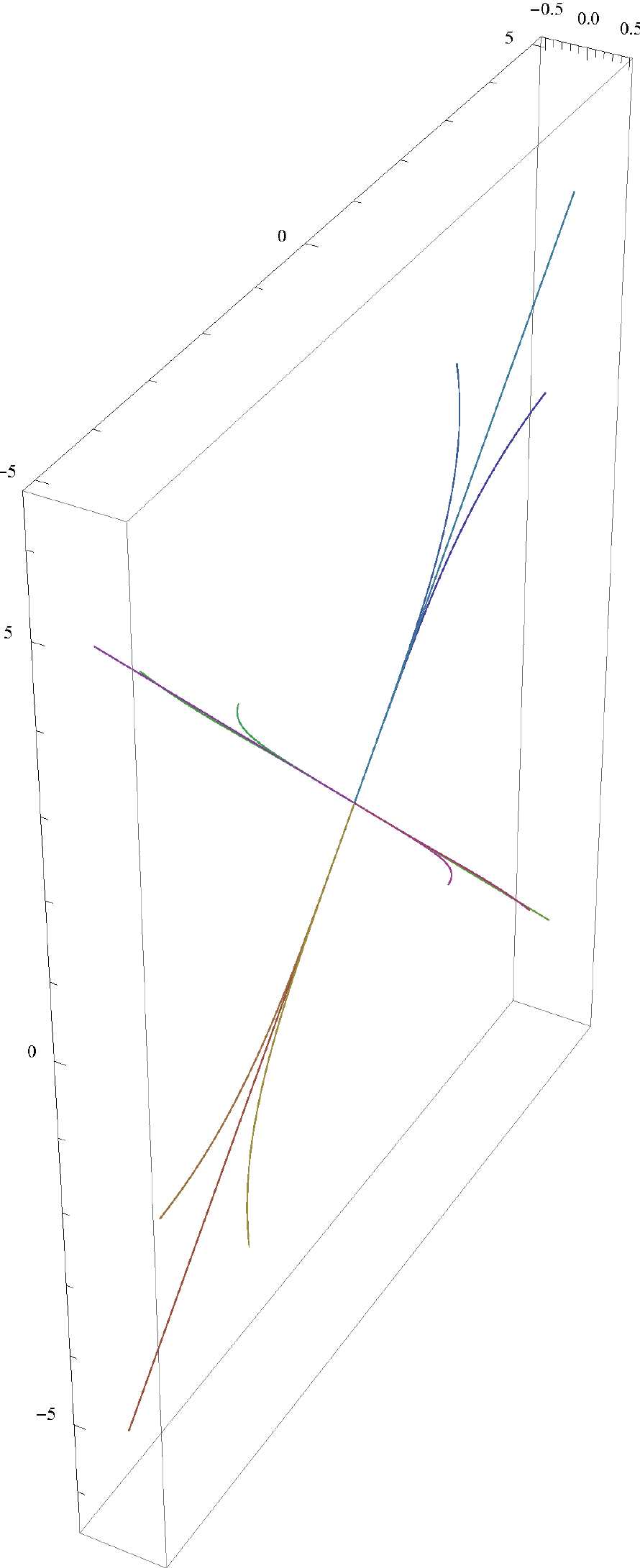}
 \includegraphics[width=.3\textwidth]{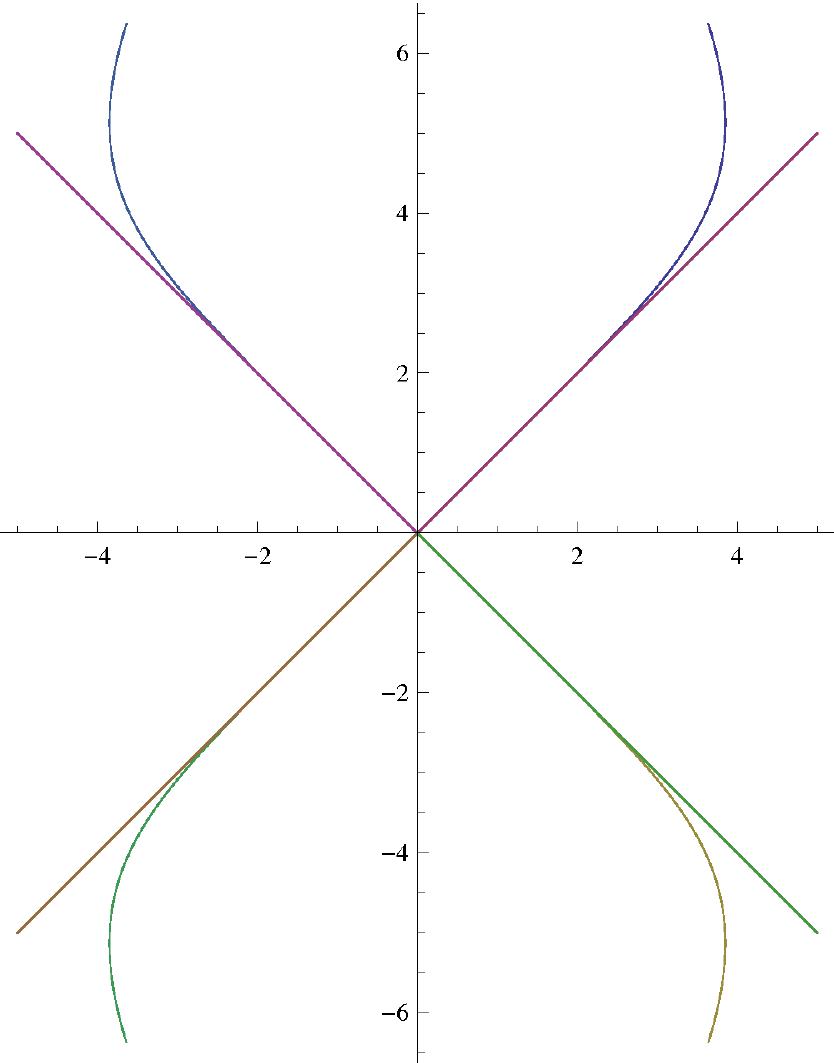}
 \includegraphics[width=.3\textwidth]{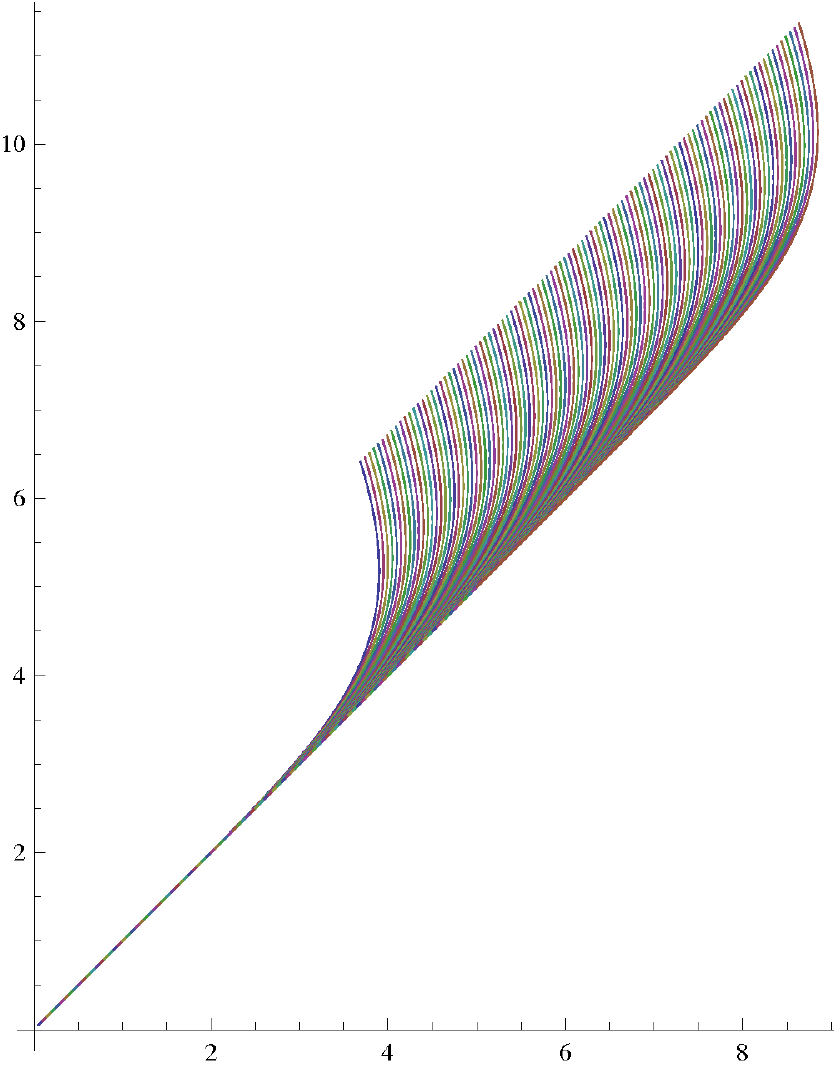}
 \includegraphics[width=.2\textwidth]{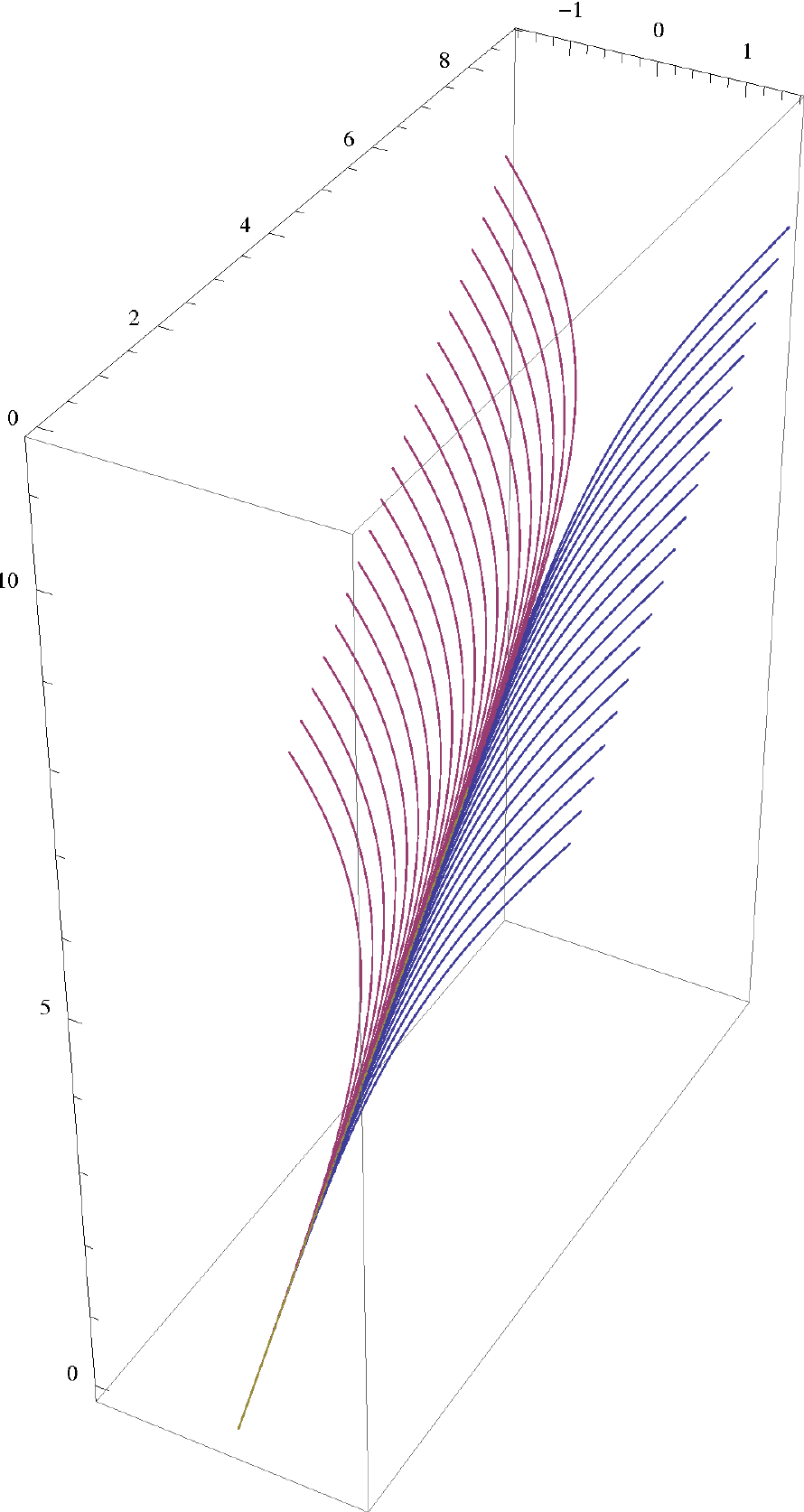}
\caption{\label{F12IV11.1} The ``extremal" bifurcating geodesics through $(0,0,0)$ (first figure), and their $t-y$ plane projections (second figure), $\gamma=1/2$. In each quadrant of the second figure, every point between the two curves is the $t-y$ projection of a point lying on a null geodesic obtained by following one of the diagonals, and then branching off to one of the (translated) extremals; see the third figure. The null-geodesics so obtained thread, in each quadrant, two hypersurfaces in space-time intersecting at an edge $\{x=0, t=y\}$, as seen in the fourth figure.}
\end{center}
\end{figure}
\begin{figure}[ht]
 \begin{center}
 \includegraphics[width=.4\textwidth]{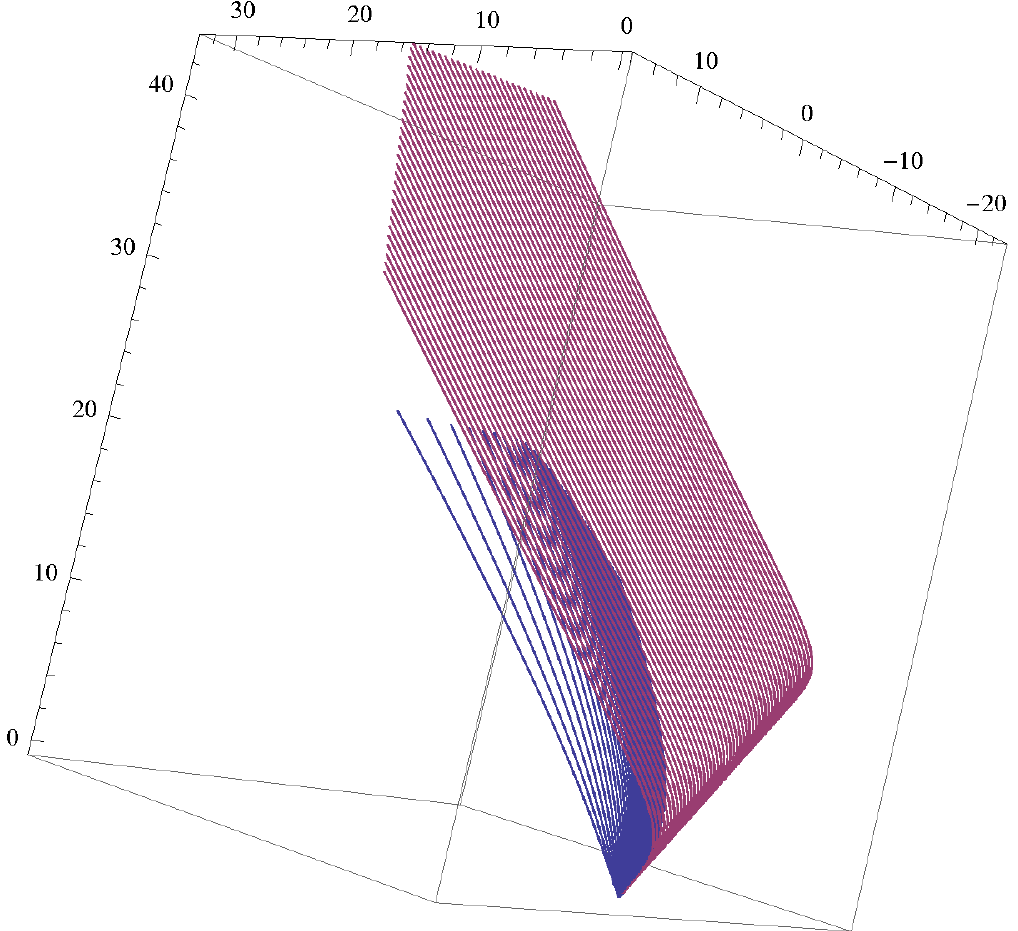}
\caption{\label{F12IV11.2} Some of the geodesics issued from the origin for the metric $g(0)$ with $\gamma=3/8$. The lower blue sheet corresponds to the usual unique geodesics with initial tangent vector $\dot x\ne 0$, the bordeaux-ones are the bifurcating ones with $\dot x(0)=0$. Each family forms a two-dimensional surface in space-time.}
\end{center}
\end{figure}

We now wish to study how the geodesics behave for the metrics with $\eps \neq 0$ and, in particular, in which sense they approximate the geodesics above for $\eps = 0$. Quite generally for $\epsilon \ne 0$ the equations for null geodesics become
\bean
 &
 \dot t = C_t \left(1 + \epsilon+\frac{1} 2 (\epsilon^2 + x^2)^{\frac{1+\gamma}2}\right)
 \;,
 \quad
 \dot y = C_y \left(1 - \frac{1} 2 (\epsilon^2 + x^2)^{\frac{1+\gamma}2}\right)
 \;,
 &
\\
 &
 \dot x^2 = C_t^2(1+\epsilon)- C_ y^2 + \frac {C_t^2 + C_y^2}2 (\epsilon^2 + x^2)^{\frac{1+\gamma}2}
 \;.
&
\eeal{11IV11.5x}
It is convenient to normalise the geodesics so that $\dot x^2 + \dot y^2 =1$ at $x=0$. Writing
$$
 \dot x(0) = \sin \phi\;,
 \quad
 \dot y(0) = \cos \phi
 \;,
$$
this leads to
\bel{15IV11.1}
 C_t^2
=\frac{   1- \frac 12 |\epsilon |^{\gamma+1
} \sin ^2\phi   }{\left(1-\frac 12|\epsilon |^{\gamma +1} \right)
   \left(1+  \epsilon + \frac 12 |\epsilon |^{\gamma +1}\right)}
 \;,
 \quad C_y = \frac{\cos\phi}{  1-\frac 12 |\epsilon|^{1+\gamma}}
 \;.
\ee
%

We then have

\begin{Proposition}
The null cones of points with $x = 0$ for the metric $g_{\epsilon}$ are smaller than those of $g$ for $\epsilon>0$ and larger if $\epsilon < 0$.
\end{Proposition}
\begin{proof}
Letting $x(0) = 0$, it follows from~(\ref{15IV11.1}) and the null geodesic equations that we have
\[
\left( \frac{dt}{ds}(0) \right)^2 = \left( 1+  \epsilon + \frac{1}{2} |\epsilon |^{\gamma +1} \right)
\frac{1-\frac{1}{2} |\epsilon |^{\gamma +1} \sin^2\phi}{1-\frac{1}{2} |\epsilon |^{\gamma +1}}  = 1 + \eps + o(\eps).
\]
Therefore, for future-directed null geodesics, we have $\frac{dt}{ds}(0) = 1 + \frac{1}{2} \eps + o(\eps)$. It follows that a null tangent vector at $x=0$ with fixed values of $\dot{x}, \dot{y}$ will have a larger value of $\frac{dt}{ds}(0)$ if $\eps > 0$ and a smaller value if $\eps < 0$.
\end{proof}

Moreover, contrary to the behaviour of geodesics with $\eps = 0$, we have the following.

\begin{Proposition}
A null geodesic with $x(0) = 0$ and $\frac{dx}{ds}(0) = 0$ with respect to the metric $g_{\eps}$ with $\eps \neq 0$ has $x(s) = 0$ for all $s$.
\end{Proposition}
\begin{proof}
Let us write
$$
 \dot x^2 = a +b (\epsilon^2 + x^2)^{\frac{1+\gamma}2}
 \;.
$$
Setting $w=|\epsilon|^{1+\gamma}$ we have
\beaa
a &= & \frac{-\left(\epsilon \left(w^2-2 w+4\right)+w^2+4\right) \cos (2 \phi )+\epsilon
   \left(w^2-6 w+4\right)+w^2-8 w+4}{2 (w-2)^2 (2 \epsilon +w+2)}
\\
 & = &
  \sin ^2 \phi - \frac{1}{4} w (1 + \cos^2 \phi)+O\left(w^2\right) +\epsilon
   \left(\frac{w}{4}+O\left(w^2\right)\right)+O\left(\epsilon ^2\right)
\;,
\eeaa
\beaa
b & = & \frac{2 (2 \epsilon +w+2) \cos ^2\phi +(w-2) \left(w \sin ^2\phi-2\right)}{(w-2)^2 (2
   \epsilon +w+2)}
\\
 & = &
 \frac{1}{2} (1 + \cos^2 \phi)+\frac{1}{4} w (3 \cos^2\phi - 1)+O\left(w^2\right)
\\
 &&+\epsilon  \left(-\frac{1}{2}+\frac{1}{4} \left(\sin ^2\phi +1\right) w+O\left(w^2\right)\right)+O\left(\epsilon ^2\right)
\;.
\eeaa
At $\sin\phi =0$ we find
\beaa
a &= &  -\frac{4 (\epsilon +2) w}{(w-2)^2 (2 \epsilon +w+2)} = -
 b w
\;.
\eeaa
This gives, for $\phi $ equal zero or $\pi$,
\bea
 \dot x^2 =\frac{4 (\epsilon +2)  }{(2-|\epsilon|^{1+\gamma})^2 (2 \epsilon +|\epsilon|^{1+\gamma}+2)}\left((\epsilon^2 + x^2)^{\frac{1+\gamma}2}-|\epsilon|^{1+\gamma}\right)
 \;.
\eeal{15IV11.15}
By uniqueness of solutions, when $\epsilon\ne 0$ the only solution with $x(0)=0$ is clearly $x(s)=0$ for all $s$.
\end{proof}

\begin{Remark}
This result could have been anticipated, as the plane $\{x=0\}$ is totally geodesic with respect to $g_{\eps}$, with a manifestly flat induced metric.
\end{Remark}

Finally, we study properties of the geodesics of the metric $g$ with $\dot{x}(0) \neq 0$. Recall that, when $\epsilon=0$, we have
\begin{equation}
\label{xphi}
 \dot x^2 = \sin^2 \varphi + \frac{1 + \cos^2 \varphi}{2} |x|^{{1+\gamma}}
\;.
\end{equation}
Consider the set of solutions $x(s)$ of this equation with initial value $\dot x(0)=\sin \varphi$, $\varphi\in (0,\pi)$,
set
$$
 x_-(s)=\inf_{\varphi\in (0,\pi/2)}|x(s)|
\;.
$$

By definition, any null geodesic, parameterised as above, that accumulates at $x=0$ and for which
$$
  |x(s)| \le x_-(s)
$$
must belong to the family of bifurcating null geodesics that reach the hyperplane $\{x=0\}$. We claim that
$$
 \mbox{$x_-(s)>0$ for $s>0$. }
$$
To see this, let $x_0$ by the positive non-trivial solution through the origin of
$$
 \dot x_0^2 =   \frac 12  |x_0|^{{1+\gamma}}
 \;.
$$
Then
\beaa
 \dot x^2 -\dot x_0^2 &= &
  \sin^2 \varphi + \frac{1}{2} \cos^2 \varphi \, |x|^{{1+\gamma}}
 + \frac 12 (|x|^{{1+\gamma}}-|x_0|^{{1+\gamma}})
\\
 & \ge & \frac 12 (|x|^{{1+\gamma}}-|x_0|^{{1+\gamma}})
 \;.
\eeaa
For very small $s$ we have $|x|\sim |\sin \varphi|s$ while $x_0\sim s^{2/(1-\gamma)}$, hence the right-hand side is positive for small $s$. It is then standard to show that \emph{all non-trivial solutions passing through $x=0$\/} satisfy
$$
 |x(s)| \ge x_0(s)
 \;.
$$
In particular $x_-$ is \emph{non-trivial}
$$
 x_- \ge x_0
\;,
$$
as claimed. See Figures~\ref{F23IV11.1} and \ref{F23IV11.2}.
\begin{figure}[ht]
 \begin{center}
 \includegraphics[height=3.8cm]{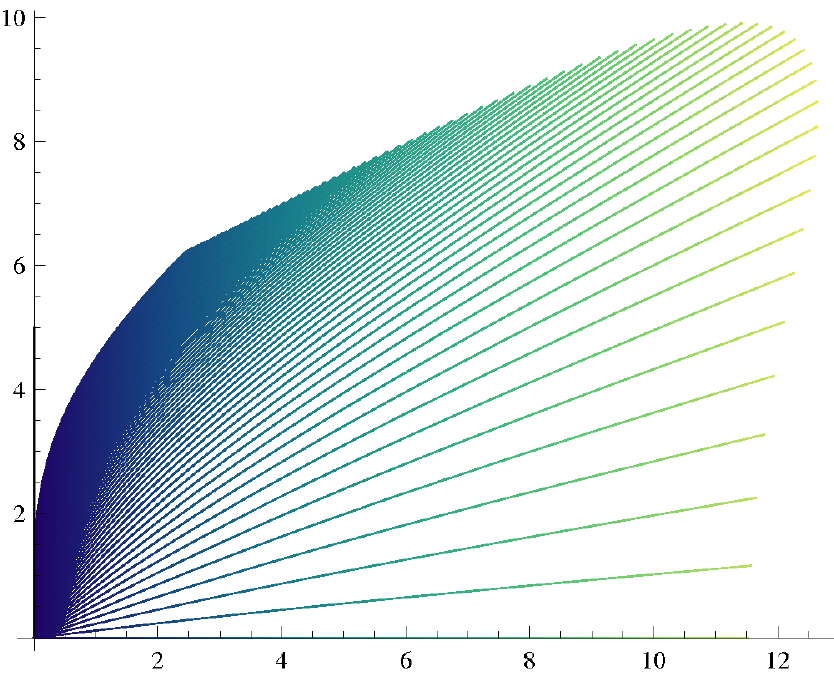}
 \includegraphics[height=3.8cm]{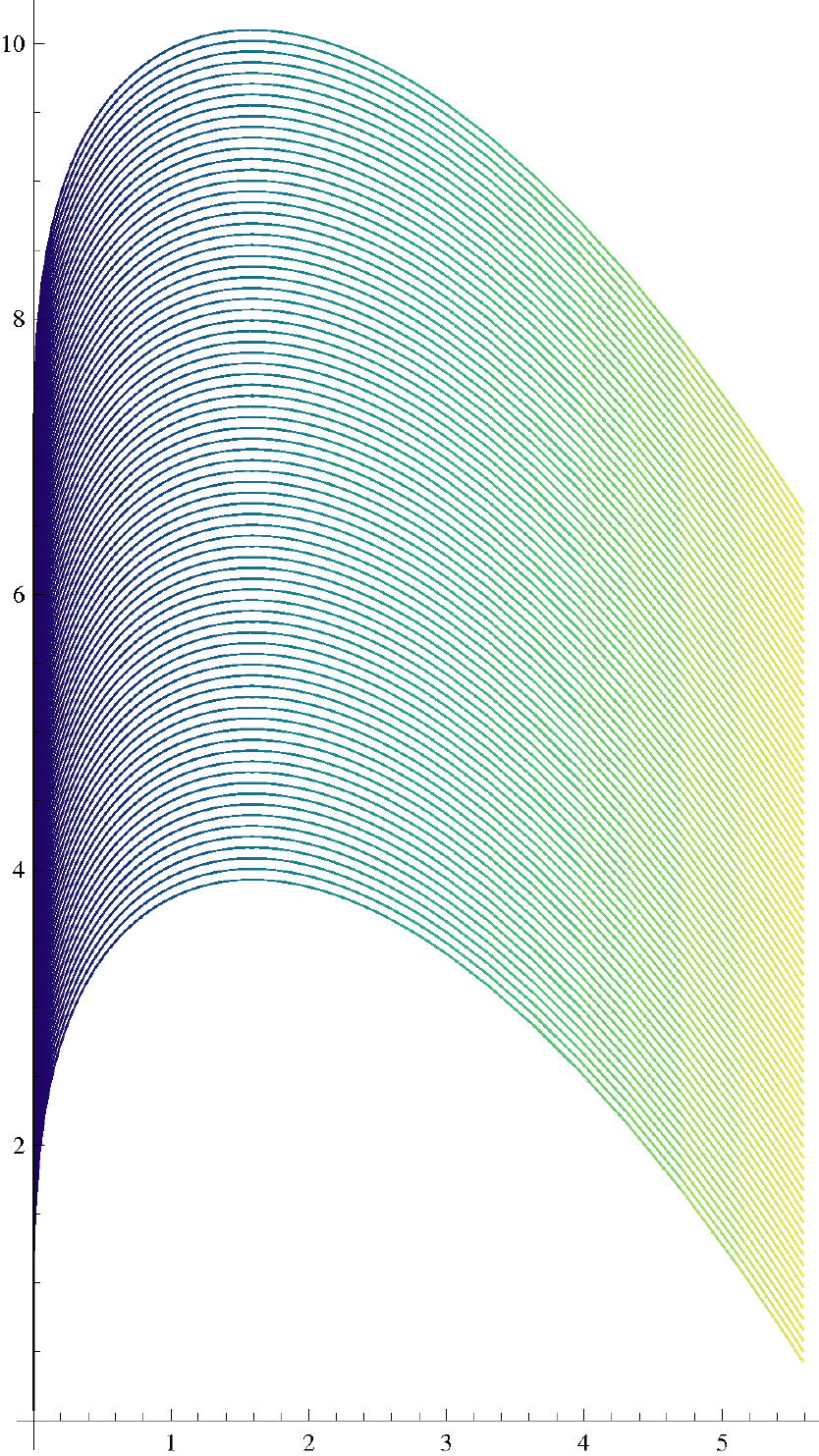}
 \includegraphics[height=3.8cm]{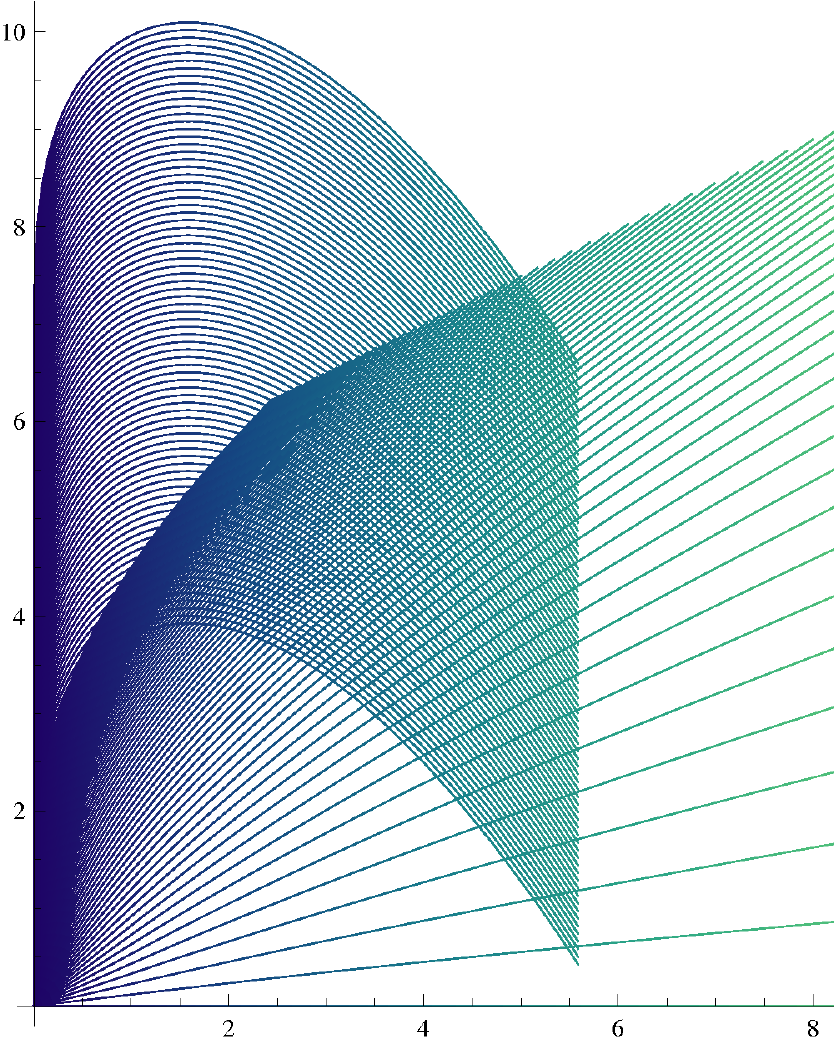}
\caption{\label{F23IV11.1} (One quarter of) the $x-y$ plane projection of the surface threaded by the non-bifurcate null geodesics (left figure), the bifurcate ones (middle figure), and both projections superimposed (right figure), with $\gamma=1/2$. The bifurcate null geodesics are vertical translations of each-other and continue indefinitely in the vertical direction. The full projection is obtained by reflecting across the $x$ axis, and then across the $y$ axis.}
\end{center}
\end{figure}
\begin{figure}[ht]
 \begin{center}
 \includegraphics[height=4cm]{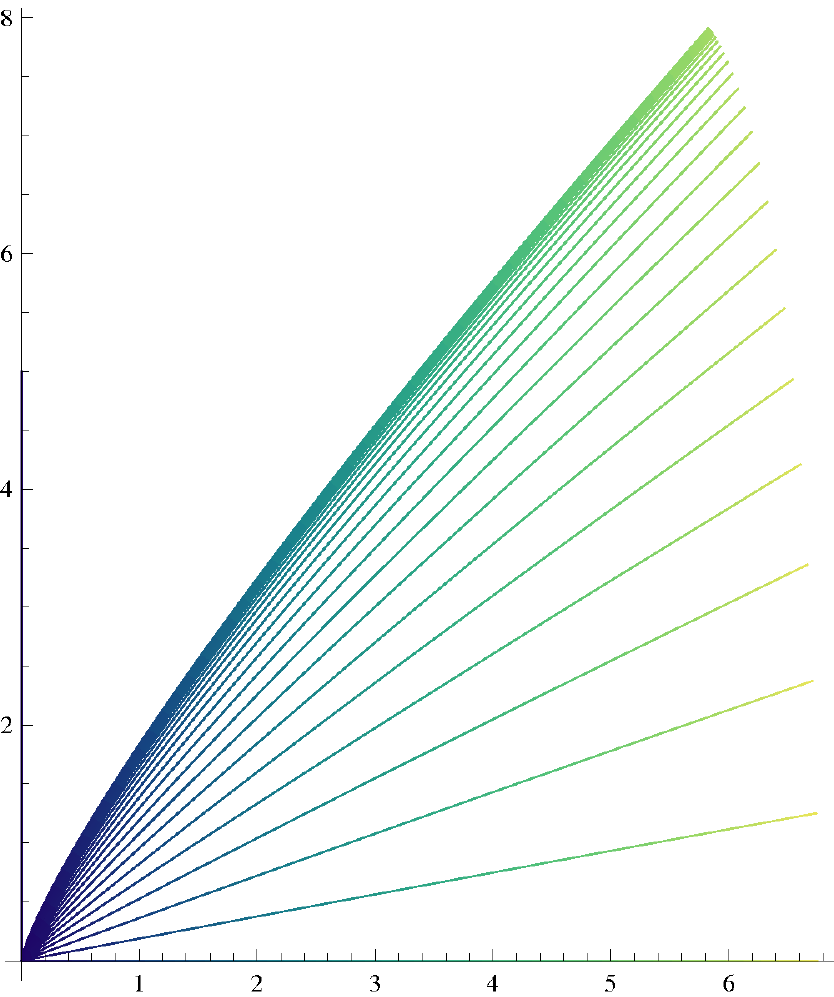}
 \includegraphics[height=4cm]{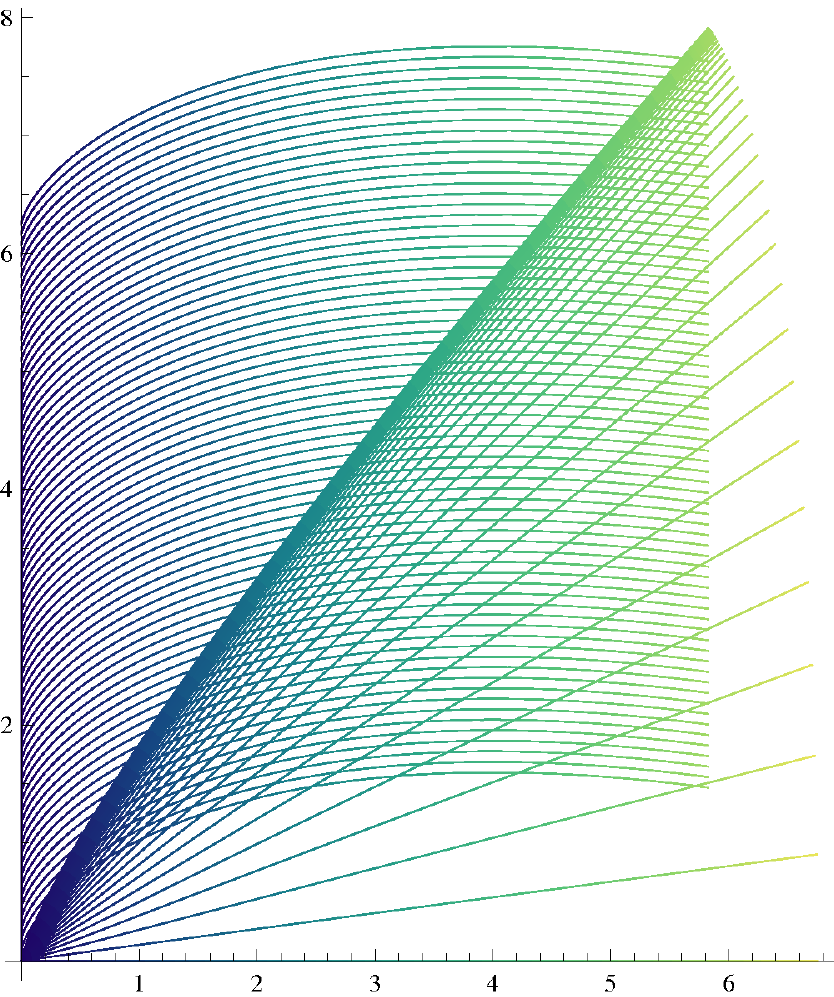}
\caption{\label{F23IV11.2} Same as the first and third figures in Figure~\ref{F23IV11.1} but with $\gamma=-1/2$.}
\end{center}
\end{figure}

The remainder of our study of the current example aims to provide  more control of the function $x_-$ above, see Corollary~\ref{C27III12.1} below:

\begin{Lemma}
Let $\varphi \in \left( 0, \frac{\pi}{2} \right)$, and $x$ denote the solution of~(\ref{xphi})
with the initial conditions $x(0) = 0$ and $\frac{dx}{ds}(0) = \sin\varphi$. Then, for sufficiently small $s > 0$, we have
\begin{equation}
\fl
\frac{\sin\varphi}{\sqrt{\frac{1}{2} \left( 1 + \cos^2\varphi \right)}} \sinh \left( \sqrt{\frac{1}{2} \left( 1 + \cos^2\varphi \right)} s \right)
\le x(s)
\le \left( \sin\varphi \right) s + \frac{1}{8} \left( 1 + \cos^2\varphi \right) s^2.
\label{comparisoninequality}
\end{equation}
\end{Lemma}

\proof
Since $x(0) = 0$, there exists $s_0 = s_0(\varphi) > 0$ such that $|x(s)| \le 1$ for $0 \le s \le s_0(\varphi)$. Therefore, for such $s$, we have
\begin{equation}
|x(s)|^2 \le |x(s)|^{1+\gamma} \le |x(s)|.
\label{xcomparison}
\end{equation}
For notational simplicity, let
\[
A_{\varphi} := \sin^2\varphi, \qquad B_{\varphi} := \frac{1 + \cos^2\varphi}{2},
\]
and note that $A_{\varphi}, B_{\varphi} > 0$. Since $\frac{dx}{ds}(0) = \sin\varphi > 0$, we have
\[
\frac{dx}{ds} = \sqrt{A_{\varphi} + B_{\varphi} |x(s)|^{1+\gamma}}.
\]
Let $X_{\varphi}(s)$ and $Y_{\varphi}(s)$ be solutions of the related problems
\beaa
\frac{dX_{\varphi}}{ds} &=& \sqrt{A_{\varphi} + B_{\varphi} X_{\varphi}(s)}, \qquad X_{\varphi}(0) = 0,
\\
\frac{dY_{\varphi}}{ds} &=& \sqrt{A_{\varphi} + B_{\varphi} Y_{\varphi}^2(s) }, \qquad Y_{\varphi}(0) = 0.
\eeaa
It then follows from~\eq{xcomparison} that we have
\[
Y_{\varphi}(s) \le x(s) \le X_{\varphi}(s).
\]
Explicitly solving the differential equations for $X_{\varphi}(s), Y_{\varphi}(s)$ yields the inequality~\eq{comparisoninequality}.
\qed

\begin{Corollary}
Let
\[
s_0(\varphi) := \frac{4}{1 + \cos^2\varphi} \left[ - \sin\varphi + \sqrt{\sin^2\varphi + \frac{1}{2} \left( 1 + \cos^2\varphi \right)} \right].
\]
Then $0 < x_{\varphi}(s) \le 1$ for $0 < s \le s_0(\varphi)$.
\end{Corollary}

\begin{proof}
The inequality~\eq{comparisoninequality} is valid as long as $x(s) \le 1$. This condition is satisfied if $X_{\varphi}(s) \le 1$, which is true if $0 \le s \le s_0(\varphi)$, with $s_0(\varphi)$ as given. Moreover, $Y_{\varphi}(s) > 0$ for all $s > 0$. Therefore, the required result follows from~\eq{comparisoninequality}.
\end{proof}

\begin{Corollary}
 \label{C27III12.1}
For $0 \le s \le \min \{ s_0(\varphi), 4 \}$, we have
\[
0 \le x_-(s) \le \frac{1}{4} s^2.
\]
\end{Corollary}

\proof
For fixed $s < 4$, then
\[
\inf_{0 < \varphi < \frac{\pi}{2}} X_{\varphi}(s) = \lim_{\varphi \to 0} X_{\varphi}(s) = \frac{1}{4} s^2.
\]
The result then follows from the inequality~\eq{comparisoninequality}.
\qed

\medskip

As such, the inequality~\eq{comparisoninequality} gives an upper bound on the
function $x_-(s)$, but not as sharp a lower bound. For fixed values of
$\varphi$, however, it gives an upper and lower bound on $x_{\varphi}(s)$ in
terms of two functions that approach $s \sin\varphi$ as $s \to 0$.
}
\end{cocoExa}

\section{Domains of dependence}
 \label{CSDd}
\setcounter{equation}{0}

A set $\hyp$ will be called a \emph{topological hypersurface\/} if $\hyp$ is
embedded and if every point of $\hyp$ has a neighborhood which is a level
set of a coordinate within a continuous coordinate system on $\mcM$.

Let $\hyp$ be an \emph{achronal} topological hypersurface in a
space-time $(\mcM,g)$; by this we mean that no points on $\hyp$ are connected by a timelike curve.
In~\cite{GerochDoD,PenroseDiffTopo}
the \emph{future domain of dependence\/} $\mcDSpIg$ of $\hyp$ (denoted there
as $\mcDpS$, and which we will denote by $\mcDSpI$ when the metric is
understood) is defined as the set
of points
$p\in \mcM$ with the property that \emph{every past-directed
past-inextendible timelike curve starting at $p$ meets $\hyp$
precisely once}. The \emph{past domain of dependence\/} $\mcDSmI$ is
defined by changing \emph{past-directed past-inextendible\/} to
\emph{future-directed future-inextendible\/} above. Finally one sets
\bel{Cdd1}
  \mcDSI:= \mcDSpI\cup\mcDSmI
  \;.
\ee
This is also the definition adopted in~\cite{ChCausalityv1}.
Note that, to avoid various pathologies, both in~\cite{ChCausalityv1} and here, in
the definition we include the requirement that $\hyp$ is a topological
hypersurface.

Now, Hawking and Ellis~\cite{HE} use \emph{causal curves\/} instead of
\emph{timelike ones\/} in their definition of domain of dependence. This leads to
inessential differences for causally-plain space-times, and the arguments
presented in \cite{ChCausalityv1} carry over word for word to causally-plain
space-times with continuous metrics. On the other hand, various difficulties
appear when applying the definition based on timelike curves to
causally-bubbling space-times, and these difficulties disappear
when the definition based on causal curves is used. Therefore we develop here a
theory for continuous metrics based on causal curves.

We say that a set is \emph{acausal} if no pairs of points in this set can be connected by a causal curve.
 Given an \emph{acausal topological hypersurface\/} $\hyp$, the \emph{future domain of
dependence\/} $\mcDpScg$ of $\hyp$ is defined as the set of points
$p\in \mcM$ with the property that \emph{every past-directed
past-inextendible $g$-causal curve starting at $p$ meets $\hyp$
precisely once}. We will write $\mcDpSc$ for $\mcDpScg$ if the metric is
unambiguous. We will talk about \emph{$I$-domains of dependence\/} and
\emph{$J$-domains of dependence\/} whenever the distinction is necessary.
We emphasise that, unless explicitly
indicated otherwise, throughout the remainder of this section ``domain of
dependence'' will stand for ``$J$-domain of dependence''.
The \emph{past domain of dependence\/} $\mcDmScg$ is
defined in the obvious way, and of course
\bel{C23V11.1x}
 \mcDScg:= \mcDpScg\cup\mcDmScg\;.\ee

There exist obvious variations of the above based on $\cI$ and
on $\cJ$.
Clearly,
\bel{24V11.2}
 \mcDpScg \subset\mcDSpIg
 \;,\ee
etc.

\begin{cocoExa}
 \label{Exa21X11.0}
 {\rm
The following
example shows that the sets $\mcDpSc$ and $\mcDSpI$ are essentially different:
Let $g$ be the bubbling metric of Example~\ref{Exa19V11.1},
and let $\hyp$ be any spacelike hypersurface such that the origin lies on the edge of $\hyp$, for example
$\hyp=\{u=-\Lambda x\;, 0<x<\epsilon\}$, with $\Lambda>0 $ and $\epsilon>0$ chosen small enough so that $\hyp$ is spacelike, see Figure~\ref{F28IX11.1}.
\begin{figure}[ht]
\begin{center}
 \includegraphics[width=.4\textwidth]{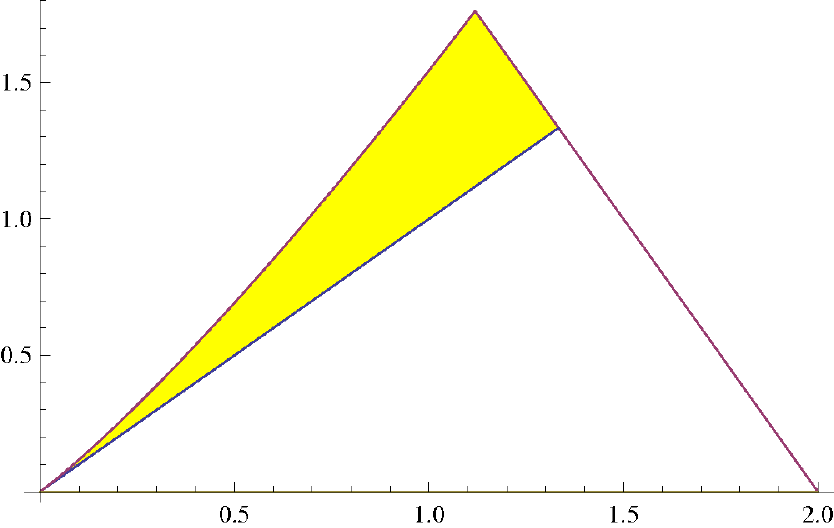}
\end{center}
 \caption{Introduce a coordinate system in which $\hyp:=\{t=0,x\in (0,2)\}$.  The $I$-domain of dependence of $\hyp$  includes the shaded bubble region, while the $J$-domain of dependence, which is the region between the $x$-axis and the lower graph, does not.\label{F28IX11.1}}
\end{figure}
Then the domain of
dependence $\mcDSpI$ includes the bubble region, while $\mcDpSc$ does not.
}
\end{cocoExa}

\medskip

For $C^{1,1}$ metrics, a key property of horizons is existence of generators.
It is not clear what the right notion
of a horizon generator is in the current example: the uppermost null geodesics of the bubble in
Figure~\ref{F28IX11.1}? the lowermost one? any null geodesic within the bubble?
Recall that generators play an important role in many causality arguments, e.g. the proof of uniqueness of maximal globally hyperbolic developments of vacuum initial data.
It would therefore be of interest to have a useful analogue of the definition of generators in the continuous setting.

\begin{cocoExa}
 \label{Exa21X11.1}
 {\rm
A simple gluing construction using two copies of the metric of Example~\ref{Exa19V11.1} provides a metric where the future bubble of $p$ is refocussed to a past bubble of $q$ as in Figure~\ref{F21X11.1}.
\begin{figure}[ht]
\begin{center}
 \includegraphics[width=.3\textwidth]{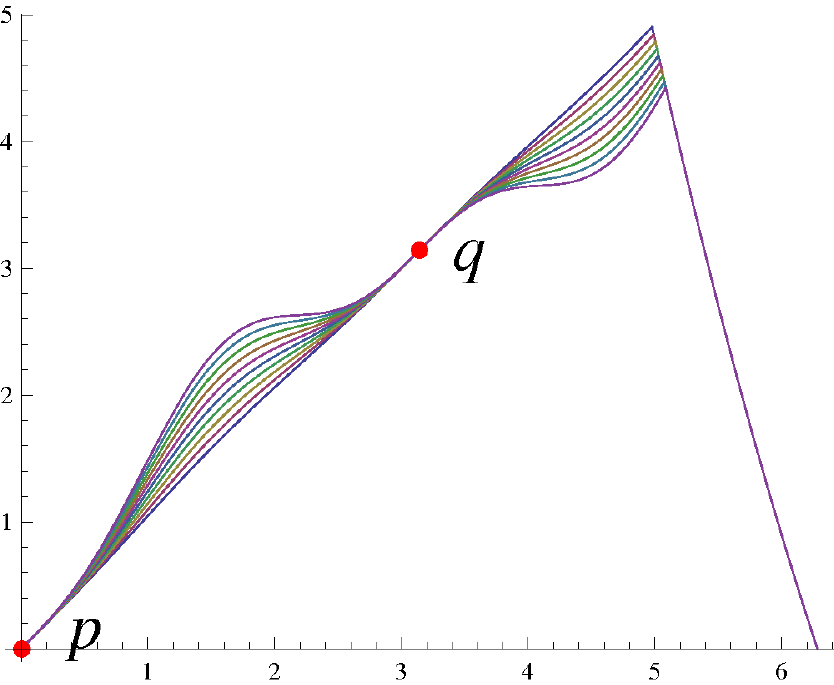}
\end{center}
 \caption{The future bubble of $p$ is refocussed to a past bubble of $q$.
   The part of the bubble between $p$ and $q$ will be referred to as the ``first bubble". The $J$-domain of dependence of $\hyp$ is the region strictly below all the graphs. Any point in the first bubble lies on the image of a causal curve obtained by following backwards in time the lowest null curve of the second bubble (which forms part of the boundary of $\mcDpSc$) until $q$ is reached, and then one of the causal curves filling the first bubble.\label{F21X11.1}}
\end{figure}
As in Example \ref{Exa21X11.0} we let $\hyp$ be any spacelike hypersurface the closure of which contains the origin, with the origin not in $\hyp$; in Figure~\ref{F21X11.1} this is taken to be the open interval $(0,2\pi)$ lying on the horizontal axis. Recall that generators of a Cauchy horizon are constructed by taking sequences of null curves which do not meet the partial Cauchy surface and which accumulate at some point of the horizon. In this example any point lying on the first bubble will lie on the image of such a curve.
}
\end{cocoExa}

We start our analysis of $\mcDSc$ with:

\begin{Proposition}\label{CPdd1}
Let $p\in \mcDpSc$, then
$$J^-(p)\cap J^+(\hyp)\subset \mcDpSc\;.$$
\end{Proposition}

\proof Let $q\in J^-(p)\cap J^+(\hyp)$, thus there exists a
past-directed causal curve $\gamma_0$ from $p$ to $q$.
Let $\gamma_1$ be a past-inextendible causal curve
starting at $q$. The curve $\gamma:=\gamma_0\cup\gamma_1$ is a
past-inextendible past-directed causal curve starting at $p$,
thus $\gamma$ meets $\hyp$ precisely once at some point $r\in\hyp$.
Suppose that $\gamma$ passes through $r$ before passing through
$q$. Since $q\in J^+(\hyp)$, there exists a future directed causal curve $\gamma_2$ from $\hyp$ to $q$.
Then the curve obtained following $\gamma_2$ from $\hyp$ to $q$, and then $\gamma_0$ (backwards) from $q$ to $r$ would be a future directed causal curve from $\hyp$ to $\hyp$,
contradicting acausality of $\hyp$. This shows that $\gamma$ must
meet $\hyp$ after passing through $q$, hence $\gamma_1$ meets
$\hyp$ precisely once.
\qed

\medskip

Let $\hyp$ be achronal. Recall that a set $\mcO$ is said to form a
\emph{one-sided future neighborhood\/} of $p\in \hyp$ if there exists an
open set $\mcU\subset \mcM$ such that $\mcU$ contains $p$ and
$$\mcU\cap J^+(\hyp)\subset \mcO\;.$$

Given an acausal topological hypersurface $\hyp$, it immediately follows from the definition of $\mcDpSc$ that
\bel{24V11.1}
 \hyp \subset \mcDpSc
 \;.
\ee
In particular $\mcDSc$ is never empty. It turns out that equality never occurs for differentiable spacelike
hypersurfaces $\hyp$:

\begin{Proposition}\label{PCdd0}
Let $\hyp$ be an acausal spacelike $C^1$-hypersurface in $\mcM$.
 Then $\mcDSc$ is open.
\end{Proposition}

We start with a lemma:

\begin{Lemma}
\label{LPCdd0}
Let $\hyp$ be an acausal spacelike $C^1$-hypersurface in $\mcM$.
 Then
\begin{enumerate}
\item $\mcDpSc\ne \hyp$.
\item For any $q\in \hyp
\cap \cI^-(\mcDpSc)$
the set $\mcDpSc$ forms a one-sided future
neighborhood of $q$.
\end{enumerate}
\end{Lemma}

\proof
1. Let $q\in \hyp$ and let $\mcU$ be a cylindrical neighborhood of $q$ with local coordinates $\{x^\mu\}$. Set
$$
 \mcU_\hyp=\{ (x^i): \mbox{the coordinate curve $t\mapsto (t,x^i)$ intersects $\hyp$}\}\subset \{x^0=0\}
 \;.
$$
Then $\mcU_\hyp$ is an open subset of $\R^n$
by the invariance-of-domain theorem, and $\hyp$ is a graph of a continuous function over $\mcU_\hyp$. Differentiability of $\hyp$, spacelikeness and the implicit function theorem show that the graphing function of $\hyp$ is differentiable.
Performing a linear Lorentz coordinate transformation so that the unit normal to $\hyp$ at the origin becomes $\partial_0$, and passing to a subset of $\mcU$ if necessary, we obtain a cylindrical neighborhood in which $\hyp$ remains a graph of a differentiable function, say $f$, with the coordinate-length of the gradient $|\partial f|_\delta$ smaller than $1/16$. Thus $f$ is Lipschitz-continuous with Lipschitz constant less than $1/16$. As $f(\vec 0)=0$, in each coordinate ball centred at the origin we have
\bel{24V11.5}
 |f(\vec x)|\le \frac 1{16} |\vec x|
  \;.
\ee
Choose a $\rho>0$ small enough so that
$$
 \mcU_\rho:= (- \rho, \rho)\times B(\rho) \subset \mcU
 \;.
$$
It follows from \eq{24V11.5} that $\hyp$ separates $\mcU_\rho$ in two connected components.

Let $\gamma$ be an inextendible causal curve in space-time meeting $\mcU_{\rho}$
in a connected set, then either $\gamma$ lies entirely in one connected component of $\mcU_\rho\setminus\hyp$, or $\gamma$ meets $\hyp$. In the former case $\gamma$ exits $\mcU_{\rho}$
through $(-\rho,\rho)\times \partial B(\rho)$.

Since the coordinate-slopes of the light-cones are bounded below by $1/2$, at each differentiability point of $\gamma$ we have
\bel{24V11.3}
 \left|\frac{d x^0}{ds}\right| \ge \frac 12 \left|\frac{d \vec x }{ds}\right|_\delta
 \;,
\ee
where, again, $|\cdot|_ \delta$ denotes coordinate length. Causality of $\gamma$ implies that $x^0$ is strictly monotonous on
$\gamma$, hence $\gamma$ can be parameterised by $x^0$. Writing $\gamma(x^0)=(x^0,\vec \gamma (x^0))$, Eq.~\eq{24V11.3} implies
$$
 \left|\frac{d \vec \gamma }{dx^0}\right|_\delta \le 2
 \;.
$$
Suppose that $\gamma(0)\in \mcU_{\rho/4}\cap J^+(\hyp)$,
hence $\vec \gamma(0) \in B(\rho/4)$, and it follows from the last inequality that $\gamma$ will not reach
the boundary $(-\rho,\rho)\times \partial B (\rho)$ for $|x^0|$ smaller than $\frac 38 \rho> \frac 14 \rho$. On the other hand, the graphing
function $f$ does not exceed $\pm \frac 1{16} \rho$. So $\gamma$ must intersect $\hyp$ before exiting $\mcU_\rho$.
We have thus shown that
\bel{25V11.6}
\mcU_{\rho/4}\cap J^+(\hyp) \subset \mcDpSc
 \;.
\ee
In particular, point 1. is established.

2. Since $\mcU_{\rho/4}$ is open,
point 2. follows from \eq{25V11.6}. Alternatively, let $q\in \cI^-(\mcDpSc)\cap\hyp$;
this set is non-empty by point 1, thus there exists $p\in
\mcDpSc\setminus\hyp$ such that $q\in \cI^-(p)$. Then
$\cI^-(p)\cap J^+(\hyp)$ is a future one-sided neighborhood of
$q$ which is contained in $\mcDpSc$ by Proposition \ref{CPdd1}.
\qed

\medskip

We can now pass to the
\medskip

{\noindent\sc Proof of Proposition~\ref{PCdd0}:}
By Lemma~\ref{LPCdd0} the domain of dependence $\mcDSc$ contains an open neighborhood $\mcU$ of $\hyp$.

Suppose that the result is false.
Changing time-orientation if necessary there exists a point $q\in \mcDpSc$,
a sequence of points $q_i\to q$ and a sequence of inextendible causal curves $\gamma_i$
through $q_i$ that do not meet $\hyp$.
Let $\gamma$ be an accumulation curve of the $\gamma_i$'s.
Then $\gamma$ is an inextendible causal curve through $q$, and hence meets $\hyp$.
But then the $\gamma_i$'s must intersect $\mcU\subset\mcDSc$ for $i$ sufficiently large enough, and hence must meet $\hyp$.
This gives a contradiction, and proves the Proposition.
\qed

\medskip

The next theorem shows that acausal topological
hypersurfaces can
be used to produce globally hyperbolic space-times:

\begin{Theorem}\label{CTdd1}
Let $\hyp$ be a spacelike acausal $C^1$-hypersurface
in $(\mcM,g)$. Then $\mcDSc$ equipped
with the metric obtained by restriction from $g$ is globally
hyperbolic.
\end{Theorem}

\proof
We start by noting that, by Proposition~\ref{PCdd0}, $\mcDSc$ is an open subset of $\mcM$, hence a manifold.

Suppose that $\mcDSc$
is not strongly causal. Then there exists $p\in \mcDSc$ and a
sequence $\gamma_n \colon \R\to\mcDSc$ of inextendible causal curves which
exit the $h$-distance geodesic ball $B_h(p,1/n)$ (centred at $p$ and
of radius $1/n$) and reenter $B_h(p,1/n)$ again. Changing
time-orientation of $\mcM$ if necessary, without loss of generality
we may assume that $p\in I^-(\hyp)\cup \hyp$. Since the property
``leaves and reenters" is invariant under the change $\gamma_n(s)\to
\gamma_n(-s)$, there is no loss of generality in
assuming the $\gamma_n$'s to be future-directed. Finally, we
reparameterize the $\gamma_n$'s by $h$--distance, with
$\gamma_n(0)\in B_h(p,1/n)$. Then, there exists a sequence $s_n>0$
such that $\gamma_n(s_n)\in B_h(p,1/n)$, with $\gamma_n(0)$ and
$\gamma_n(s_n)$ lying on different connected components of
$\gamma\cap B_h(p,1/n)$.

Let $\mcO$ be a cylindrical neighborhood of $p$, as in
Definition~\ref{D9IV11.1}, p.~\pageref{D9IV11.1}, and let $n_0$ be large enough so that
$\overline {B_h(p,1/n_0)} \subset \mcO$. Note that the local
coordinate $x^0$ on $\mcO$ is monotonous along every connected
component of $\gamma_n\cap \mcO$ which implies, for $n\ge n_0$, that
any causal curve which exits and reenters $B_h(p,1/n)$ has also to
exit and reenter $\mcO$. This in turn guarantees the existence of an
$\epsilon>0$ such that $ s_n>\epsilon$ for all $n\ge n_0$.

Let $\gamma$ be an accumulation curve
through $p$ of the $\gamma_n$'s, passing to a subsequence if necessary the $\gamma_n$'s
converge uniformly to $\gamma$ on compact subsets of $\R$. The curve
$\gamma$ is causal and meets $\hyp$, which implies that
there exist $s_\pm\in\R$ such that $\gamma(s_-)\in \cI^-(\hyp)$ and
$\gamma(s_+)\in \cI^+(\hyp)$. Since $\hyp$ is acausal and $\gamma$ is
future directed we must have $s_-<s_+$. Since the $\cI^\pm(\hyp)$'s are open, and
since (passing to a subsequence if necessary)
$\gamma_n(s_\pm)\to\gamma(s_\pm)$, we have $\gamma_n(s_\pm)\in
\cI^\pm(\hyp)$ for $n$ large enough.

Note that stable causality is now violated at $\gamma(s_-)$, so replacing $p$
by $\gamma(s_-)$ if necessary we can without loss of generality assume that
$p\in \cI^-(\hyp)$.

Now, we claim that $s_n \le s_+$ for
$n$ large enough: Otherwise, for $n$ large, we could follow
$\gamma_n$ to the future from $\gamma_n(s_+)\in \cI^+(\hyp)\subset I^+(\hyp)$
to $\gamma_n(s_n)\in I^-(\hyp)$, which is not possible if $\hyp$ is
achronal.

It follows that there exists $ s_*\in \R$ such that, passing again to
a subsequence if necessary, we have $s_n\to s_*$. Note that
$\gamma_n(s_*)\to p$ and that $s_*\ge \epsilon$. Since
$\gamma_n|_{[0,s_*]}$ converges uniformly to $\gamma|_{[0,s_*]}$, we
obtain an inextendible periodic causal curve $\gamma'$ through $p$
by repetitively circling from $p$ to $p$ along
$\gamma|_{[0,s_*]}$.
This is clearly incompatible with the fact that $\gamma'$ has to meet $\hyp$, and that $\hyp$ is acausal, and ends the proof of stable causality.

To finish the proof, we need to prove compactness of the sets
of the form
$$J^+(p)\cap J^-(q)\;,\quad p,q\in \mcDSc\;.$$
If $p$ and $q$ are such that this set is empty or equals $\{p\}$
there is nothing to prove. Otherwise, consider a sequence $r_n \in
J^+(p)\cap J^-(q)$. One of the following is true:
\begin{enumerate}
\item we have $r_n\in J^-(\hyp) $ for all $n\ge n_0$, or
\item there exists a subsequence, still denoted by $r_n$, such that $r_n\in
J^+(\hyp)$.
\end{enumerate}
In the second case we change time-orientation, pass to a
subsequence, rename $p$ and $q$, reducing the analysis to the first
case. Note that this leads to $p\in J^-(\hyp)$.

By definition, there exists a future directed causal curve
$\hat\gamma_n$ from $p$ to $q$ which passes through $r_n$,
\bel{Cedd1}
 \hat \gamma_n(s_n)=r_n\;.
 \ee
Let $\gamma_n$ be any $\distb$-parameterized, inextendible future
directed causal curve extending $\hat\gamma_n$, with
$\gamma_n(0)=p$. Let $\gamma$ be an inextendible accumulation curve
of the $\gamma_n$'s, then $\gamma$ is a future inextendible causal
curve through $
 p\in \mcDmSc$. Thus there
exists $s_+$ such that $\gamma(s_+)\in \cI^+(\hyp)$.
Passing to a
subsequence, the $\gamma_n$'s converge uniformly to $\gamma$ on
$[0,s_+]$, which implies that for $n$ large enough the
$\gamma_n|_{[0,s_+]}$'s enter $\cI^+(\hyp)$.
This, together with
acausality of $\hyp$, shows that the sequence $s_n$ defined by
\eq{Cedd1} is bounded; in fact we must have $0\le s_n\le s_+$.
Eventually passing to another subsequence we thus have $s_n\to
s_\infty$ for some $s_\infty\in \R$. This implies
$$
 r_n\to \gamma(s_\infty)\in J^+(p)\cap J^-(q)\;,
$$
which had to be established.
 \qed

\bigskip

We continue with an argument inspired by~\cite{NavarroMinguzzi}:

\begin{Theorem}
 \label{T9IX11.1}
 Let $\hyp$ be an acausal spacelike $C^1$ hypersurface
 in $\mcM$ and let $g$ be continuous on $\mcM$. There exists a smooth metric $\hg\succ g$ on $\mcDpSc$ such that $(\mcDpSc,\hg)$ is globally hyperbolic with Cauchy surface $\hyp$.
\end{Theorem}

Here $
 (\mcDpSc,\hg)$ is viewed as a space-time with boundary $\hyp$,
and global
hyperbolicity above has to be understood within the class of space-times having
the Cauchy surface as a boundary.

\begin{Remark}
 \label{R12XI11.1}
{\rm
It further follows from the proof that there exists a sequence of smooth metrics $\hg_n\succ g$ which converge locally uniformly to $g$ such that $(\mcDpSc,\hg)$ is globally hyperbolic with Cauchy surface $\hyp$.
}
\end{Remark}

\bigskip

\proof
Let $p_0\in \hyp$, let $h$ be an auxiliary complete Riemannian metric on $\mcDpSc$ and let $B_r\subset \mcDpSc$ be a closed $h$-ball of radius $r$ centred at $p_0$. Set
$$
 K_n = J^-_g(B_n,\mcDpSc)\subset \mcDpSc
 \;.
$$
 Clearly
$$
 K_n \subset K_{n+1}
 \;,
 \quad
 \cup_n K_n = \mcDpSc\;,
  \quad
   \cup_n (K_n\cap \hyp) = \hyp
   \;.
$$
We start by showing  that $K_i$ is compact for all $i\in \N$. For this, let $r_n$ be a sequence of points in $K_i$. Then there exists a sequence of past directed inextendible causal curves $\gamma_n$, parameterized by $h$-distance, starting at $\gamma_n(0)=p_n\in B_i$, ending at $q_n=\gamma_n(t_n)\in \hyp$, and passing through $r_n=\gamma_n(s_n)$, for some $0\le s_n\le t_n<\infty$.  By the Hopf--Rinow theorem $B_i$ is compact, therefore there exists $p\in B_i$ and a subsequence, still denoted by $p_n$, such that $p_n$ converges to $p\in B_i\subset\mcDpSc$. Let $\gamma$ be an accumulation curve of the $\gamma_n$'s passing through $p$, then $\gamma$ is contained in $\mcDpSc$ and, by definition of $\mcDpSc$, the curve $\gamma$ intersects $\hyp$ at some point $\gamma(t)$. This implies that, passing to a subsequence if necessary, $t_n$ converges to $t$; in particular the sequence $t_n$ is bounded. Therefore the sequence $s_n$ is bounded as well, which implies that there exists $s_\infty\in [0,t]$ such that a subsequence, still denoted by $s_n$, converges to some $s_\infty$, whence $\gamma_n(s_n)$ converges to $\gamma(s_\infty)=:q$. Thus the sequence $q_n$ has a subsequence converging to a point $q\in K_i$, and $K_i$ is compact, as claimed.

We have also proved that $K_n\cap \hyp$ is compact for all $n$, being the intersection of a compact set with a closed set.

Let
\bel{25III12.1}
 \hg_k\succ g
\ee
be a sequence of smooth metric converging locally uniformly to $g$. We claim that there exists $k(n)$ such that every past inextendible $\hg_{k(n)}$-causal curve meeting $K_n$ meets $\hyp$.
Otherwise, passing to a subsequence if necessary, for every $k$ there exists a point $p_k \in K_n$ and a past inextendible $\hg_k$-causal curve $\gamma_k$ through $p_k$ that does not meet $\hyp$. Passing to a further subsequence if necessary, compactness of $K_n$ implies that there exists $p\in K_n$ such that $p_k$ converges to $p$. Let $\gamma$ be an accumulation curve of the $\gamma_k$'s through $p$, then $\gamma$ is $g$-causal past inextendible, and by global hyperbolicity of $\mcDpSc$ the curve $\gamma$ meets $\hyp$. But then all the $\gamma_k$'s meet $\hyp$ for $k$ large enough, a contradiction.

A similar contradiction argument can be used to show that
$$
 J^-_{\hgkn}(K_n)\cap \hyp \subset K_{n+1}
 \;,
$$
increasing $k(n)$ if necessary. This implies that, increasing $k(n)$ again if necessary, all $\hgkn$-causal past inextendible curves through $K_n$ meet $\hyp$ only once: Indeed, for all $k$ large enough the hypersurface $\hyp\cap K_{n+1}$ is spacelike for $\hg_k$, hence the intersection of any $\hg_k$-causal curve $\gamma_k$ with $\hyp\cap K_{n+1}$ is transverse, so that any such intersection point must be the end point of $\gamma_k$.

Summarising, we have shown that every past inextendible $\hgkn$-causal curve through $K_n$ meets $\hyp$ precisely once within the compact set $\hyp\cap K_{n+1}$.

Let $\theta^0$ be any smooth timelike one-form on $\mcM$.
We let $\epsilon(n) $ be any strictly positive number satisfying
$$
 \hg_{k(n)}\succ g +2\epsilon(n) \theta^0\otimes \theta^0
  \;,
   \qquad
     \epsilon \||\theta^0|^2_h\|_{L^\infty(K_n)}< 1/n
 \;.
$$
Let $f$ be any measurable function on $\mcDpSc$ such that $f(p)\le \epsilon(n_p)$, where $n_p$ is the smallest integer such that $p\in K_{n_p}$.
 Let $\hg$ be a smoothing of $g+f \theta^0\otimes \theta^0$ such that $\hg_{k(n)}\succ \hg \succ g$ on $K_n$. It is easily seen that $\hg$ has the properties claimed.
\qed

\bigskip

Recall that a
\emph{Cauchy time function\/} $t$ is a time function such that all level sets of $t$ are Cauchy.
As a straightforward corollary of Theorem~\ref{T9IX11.1} we obtain:

\begin{Theorem}
 \label{T9IX11.2}
On every domain of dependence $(\mcDpSc,g)$ with a continuous metric $g$ there
exists a smooth Cauchy time function.
\end{Theorem}

\proof
Let $\hg$ be given by Theorem~\ref{T9IX11.1}. Then there exists a smooth Cauchy time function $t$ for $(\mcDpSc,\hg)$. Since $\hg\succ g$ the function $t$ is easily seen to be Cauchy for the metric $g$ as well.
\qed

\bigskip

We finish the section by the following characterisation of domains of dependence:

\begin{Theorem}
 \label{CTdd3}
Let $\hyp$ be a spacelike acausal $C^1$-hypersurface and let $g$ be continuous on $\mcM$.
A point $p\in\mcM$ is in $ \mcDpSc$ if and only if
\bel{CeTdd3}
 \hspace{-.5cm}
 \mbox{the set }\ {J^-(p)\cap \hyp}\ \mbox{ is
 non-empty, and compact as a subset of $\hyp$.}
 \qquad
 \ee
\end{Theorem}

 \proof
For $p\in \mcDpSc$, compactness of $J^-(p)\cap \hyp $
can be established by an argument very similar to that given in
the last part of the proof of Theorem~\ref{CTdd1}. The details are
left to the reader.

In order to prove the reverse implication  we assume first that the metric is $C^2$ (in fact, Lipschitzian would suffice in view of Corollary~\ref{C23V11.1}), and that \eq{CeTdd3}
holds. Then there exists a future directed causal curve
$\gamma \colon [0,1]\to\mcM$ from some point $q\in\hyp$ to $p$. Set
$$I:= \{t\in [0,1]: \gamma(s)\in \mcDpSc\ \mbox{\ for all $s\le t$}\}\subset [0,1]\;.$$
Then $I$ is not empty, since $\mcDpSc$ forms an open neighborhood of
$\hyp$ by Proposition~\ref{CPdd1}.
The interval $I$ is open in $[0,1]$ because $\mcDpSc$ is open. In order to show that $I$ equals
$[0,1]$ set
$$
t_*:= \sup I\;.
$$
Consider any past-inextendible past-directed
causal curve $\hat \gamma$ starting at $\gamma(t_*)$. For $t<t_*$
let $\hat \gamma_t$ be a family of past-inextendible causal
push-downs of $\hat\gamma$, as in \cite[Lemma~(2.9.10)]{ChCausalityv1}, which start at $\gamma(t)$, and which
have the property that
$$
 \distb(\gamma_t(s),\gamma(s))\le  |t-t_*|
 \ \mbox{ for} \  0\le s \le 1/|t-t_*|\;.
$$
Then $\hat\gamma_t$ intersects
$\hyp$ at some point $q_t\in J^-(p)$. Compactness of
$J^-(p)\cap\hyp$ implies that the curve $t\to q_t \in \hyp$
accumulates at some point $q_*\in \hyp$, which clearly is the
point of intersection of $\gamma$ with $\hyp$. This shows that
every causal curve $\gamma$ through $\gamma(t_*)$ meets $\hyp$, in
particular $\gamma(t_*)\in \mcDpSc$. So $I$ is both open and closed in $[0,1]$,
hence $I=[0,1]$, and the result is proved for sufficiently smooth metrics.

Now, if $g$ is not Lipschitzian, the argument just given does not apply because the push-down lemma fails. However, the argument after \eq{25III12.1} shows that for any sequence of smooth metric $\{\hg_n\}_{n\in\N}$, with cones larger than those of $g$, which converges locally uniformly to $g$, the sets   ${J_{\hg_n}^-(p)\cap \hyp}$ are compact for all $n\ge n_0$ for some $n_0$. By the argument just given for $C^2$ metrics, $p\in \mcD^+_{J,g_n}(\hyp)$ for all $n\ge n_0$, which easily implies $p\in\mcD^+_{J,g }(\hyp)\equiv\mcDpSc$.
 \qed

\subsection{Remarks on wave equations}
 \label{ss12XI11.1}

One of the key properties of globally hyperbolic space-times is that the Cauchy problem for the linear wave equation on a smooth globally hyperbolic space-time has unique smooth global solutions. It is of interest to enquire what happens with the Cauchy problem when the metric is not smooth.

Let, thus, $\hyp$ be a $C^1$, acausal, spacelike hypersurface in a space-time $(\mcM,g)$. Suppose that the metric is locally Lipschitz-continuous; this guarantees that the energy inequality holds for smooth functions. Remark~\ref{R12XI11.1} shows that the metric $g$ can be approximated by smooth globally hyperbolic metrics, and it is easily seen that the approximating sequence can be chosen to be uniformly Lipschitz on compact sets. A simple density argument then shows that for any $\varphi\in H^1_\loc(\hyp)$ and for any $ \psi\in L^2_\loc(\hyp)$ the Cauchy problem
\bel{12XI11.1}
 \Box_g u =0\;,
 \qquad
 u|_\hyp= \varphi\;,
 \quad
 \partial_n u|_\hyp=\psi\;,
\ee
where $\partial_n$ denotes the derivative in a normal direction, has a unique global weak solution on $\mcD_J(\hyp)$ (compare~\cite{NicolasCIVP,HormanderCIVP}). The solution has the property that for any $C^1$ time function $t$ on $\mcD_J(\hyp)$ with level sets $\hyp_t$ we have $u|_{\hyp_t}\in H^1_\loc(\hyp_t)$, $\partial_n u|_{\hyp_t}\in L^2_\loc(\hyp_t)$. To obtain better control of the solution it seems that one needs to make further hypotheses on the metric, see e.g.~\cite{hughes:kato:marsden,MetivierParaDiff,ChBCF,TaylorIII,SmithParametrix,SmithSogge}.

Existence of solutions for a class of metrics satisfying a log-Lipschitz condition,
$$
 |f(x)-f(y)|\le C |x-y| |\log |x-y||
 \;,
$$
with a loss of derivatives, has been proved by Colombini and Lerner in~\cite{ColombiniLernerDuke}. So going beyond Lipschitz metrics is possible, perhaps for a price.%
\footnote{It should be kept in mind that existence of solutions in a Colombeau sense can be established for metrics with Colombeau coefficients under rather mild conditions on the metric~\cite{GMS,HKS}.
 Such solutions appear, however, awkward to work with.}
%

It turns out that one cannot go much further~\cite{ColombiniLernerDuke,ColombiniSpagnoloASENS}: Given any $\lambda \in (0,1)$, Colombini and Spagnolo~\cite{ColombiniSpagnoloASENS} construct two-dimensional metrics on $\R\times S^1$ of the form
$$
 g= -dt^2  + \alpha(t,x)dx^2
 \;,
$$
with a $\lambda$-H\"older continuous positive function $\alpha$ bounded away from zero, for which the wave equation has no distributional solutions near $t=0$. They also prove in~\cite{ColombiniSpagnoloASENS} existence of a function $\alpha\in \cap_{\lambda \in (0,1)} C^{0,\lambda}$ and a smooth source term $w$ for which the non-homogeneous wave equation,
$$
 \Box_g u =w\;,
$$
has no $C^1$ solutions near the origin. Finally in~\cite{ColombiniJannelliSpagnolo}, Colombini, Jannelli and Spagnolo prove existence of functions $\alpha\in \cap_{\lambda \in (0,1)} C^{0,\lambda}$ and $b\in C^\infty$ such that the Cauchy problem for the equation
$$
\Box_g u +b u =0
$$
has non-unique solutions.

These remarks give a PDE perspective to some of our results. It would be of interest to study the bubbling properties of the Colombini--Spagnolo metrics. It is tempting to raise the question, whether the lack of causal bubbles restores well-posedness of the Cauchy problem for the wave equation.

\ack
PTC is grateful to l'IHES for hospitality and support during part of the work on this paper. He was
also supported in part by the Polish Ministry of Science and
Higher Education grant Nr N N201 372736.
JDEG is grateful to the Erwin Schr\"{o}dinger Institute,
where part of this work was completed during the programme \lq\lq Dynamics of General Relativity\rq\rq.
Useful discussions with Helmut
Friedrich, Gregory Galloway and Jacek Jezierski are acknowledged.

\section*{References}

\def\cprime{$'$}

\end{document}